\pdfoutput=1
\RequirePackage{ifpdf}
\ifpdf 
\documentclass[pdftex]{sigma}
\else
\documentclass{sigma}
\fi

\usepackage{mathrsfs,mathabx}
\newcommand{\ii}{\mathrm{i}}
\newcommand{\ee}{\mathrm{e}}
\newcommand{\dd}{\mathrm{d}}

\numberwithin{equation}{section}

\newtheorem{Theorem}{Theorem}[section]

\newtheorem{rhp}[Theorem]{Riemann--Hilbert Problem}
 { \theoremstyle{definition}

 }

\begin{document}


\newcommand{\arXivNumber}{1704.04851}

\renewcommand{\thefootnote}{}

\renewcommand{\PaperNumber}{065}

\FirstPageHeading

\ShortArticleName{Rational Solutions of the Painlev\'e-II Equation Revisited}

\ArticleName{Rational Solutions of the Painlev\'e-II Equation\\ Revisited\footnote{This paper is a~contribution to the Special Issue on Symmetries and Integrability of Dif\/ference Equations. The full collection is available at \href{http://www.emis.de/journals/SIGMA/SIDE12.html}{http://www.emis.de/journals/SIGMA/SIDE12.html}}}

\Author{Peter D.~MILLER and Yue SHENG}
\AuthorNameForHeading{P.D.~Miller and Y.~Sheng}
\Address{Department of Mathematics, University of Michigan,\\ 530 Church St., Ann Arbor, MI 48109, USA}
\Email{\href{mailto:millerpd@umich.edu}{millerpd@umich.edu}, \href{mailto:shengyue@umich.edu}{shengyue@umich.edu}}
\URLaddress{\url{http://math.lsa.umich.edu/~millerpd/}}

\ArticleDates{Received April 18, 2017, in f\/inal form August 07, 2017; Published online August 16, 2017}

\Abstract{The rational solutions of the Painlev\'e-II equation appear in several applications and are known to have many remarkable algebraic and analytic properties. They also have several dif\/ferent representations, useful in dif\/ferent ways for establishing these properties. In particular, Riemann--Hilbert representations have proven to be useful for extracting the asymptotic behavior of the rational solutions in the limit of large degree (equivalently the large-parameter limit). We review the elementary properties of the rational Painlev\'e-II functions, and then we describe three dif\/ferent Riemann--Hilbert representations of them that have appeared in the literature: a representation by means of the isomonodromy theory of the Flaschka--Newell Lax pair, a second representation by means of the isomonodromy theory of the Jimbo--Miwa Lax pair, and a third representation found by Bertola and Bothner related to pseudo-orthogonal polynomials. We prove that the Flaschka--Newell and Bertola--Bothner Riemann--Hilbert representations of the rational Painlev\'e-II functions are explicitly connected to each other. Finally, we review recent results describing the asymptotic behavior of the rational Painlev\'e-II functions obtained from these Riemann--Hilbert representations by means of the steepest descent method.}


\Classification{33E17; 34M55; 34M56; 35Q15; 37K15; 37K35; 37K40}

\renewcommand{\thefootnote}{\arabic{footnote}}
\setcounter{footnote}{0}

\section{Introduction}
\looseness=-1 Rational solutions of the second Painlev\'e equation are important in several applied problems. It was discovered by Bass \cite{Bass64} that a certain Nernst--Planck model of steady electrolysis with two ions reduces to the Painlev\'e-II equation, and in \cite{RogersBS99} the special role of rational solutions was highlighted in the context of this model (see also~\cite{BassNRS10}). Johnson \cite{Johnson06} notes that rational solutions of the Painlev\'e-II equation parametrize certain string theoretic models in high-energy physics. Clarkson \cite{Clarkson09} reviews the application of rational solutions of Painlev\'e-II to the classif\/ication of certain equilibrium vortex conf\/igurations in ideal planar f\/luid f\/low (the poles of opposite residue determine the location of vortices of opposite circulation). More recently Buckingham and one of the authors \cite{BuckinghamM12} found that the collection of rational solutions of the Painlev\'e-II equation universally describes the space-time location of kinks in the semiclassical sine-Gordon equation near a transversal crossing of the separatrix for the simple pendulum. Also, Shapiro and Tater \cite{ShapiroT14} have observed a connection between the rational Painlev\'e-II solutions of large degree and the characteristic polynomials for the quasi-exactly solvable part of the discrete spectrum in a boundary-value problem for the stationary Schr\"odinger operator with a quartic potential exhibiting PT-symmetry.

In this paper, we review some of the well-known elementary properties of the rational solutions of the Painlev\'e-II equation, including three ways to represent them in terms of the solution of a Riemann--Hilbert problem. We make a new contribution by showing that two of the three representations, until now thought to be unrelated, are in fact connected with each other via a simple and explicit transformation. Then we review results on the asymptotic behavior of the rational Painlev\'e-II solutions in the large-degree limit that have been recently obtained by analyzing these Riemann--Hilbert representations. Here our aim is to simplify the results and correct them where necessary as well as to report them in a unif\/ied context.

\subsection{Basic properties of rational Painlev\'e-II solutions}
In this paper, we consider the Painlev\'e-II equation with parameter $m$ written in the form
\begin{gather}
\frac{\dd^2p}{\dd y^2}=2p^3 + \frac{2}{3}yp-\frac{2}{3}m,\qquad p=p(y),\qquad m\in\mathbb{C}.\label{eq:myPII}
\end{gather}
\subsubsection{Necessary conditions for existence of rational solutions}
Obviously if $m=0$, one has $p(y)=0$ as a particular solution. Moreover this is the only rational solution. Indeed, any nonzero rational solution $p(y)$ would necessarily admit, for some $p_0\neq 0$ and $n\in\mathbb{Z}$, a series expansion
\begin{gather}\label{expansion}
p(y)=y^n\big(p_0+p_1y^{-1}+\cdots\big)
\end{gather} convergent for suf\/f\/iciently large $|y|$, and the a dominant balance in \eqref{eq:myPII} with $m=0$ would require $n=\tfrac{1}{2}\not\in\mathbb{Z}$.

Conversely, if $m\in\mathbb{C}$ is nonzero, \eqref{eq:myPII} does not admit the zero solution. In this case, every rational solution $p(y)$ again has the form \eqref{expansion} 
for suf\/f\/iciently large $|y|$ and $p_0\neq 0$ with $n\in\mathbb{Z}$; it then follows from~\eqref{eq:myPII} that a dominant balance is achieved between the terms $\tfrac{2}{3}yp$ and $-\tfrac{2}{3}m$ yielding $n=-1$ and $p_0=m$. Therefore if $C$ is a counterclockwise-oriented circle of suf\/f\/iciently large radius,
\begin{gather}
\frac{1}{2\pi\ii}\oint\nolimits_C p(y)\,\dd y= m.\label{eq:residue-integral}
\end{gather}
Likewise if $y_0\in\mathbb{C}$ is a f\/inite pole of order $n$ of the rational solution $p(y)$, then the substitution of the Laurent series \begin{gather*}
p(y)=(y-y_0)^{-n}(p_0+p_1(y-y_0)+\cdots)
\end{gather*}
into \eqref{eq:myPII} results in a dominant balance between $p''$ and $2p^3$ yielding $n=1$ and $p_0=\pm 1$, so every pole of $p(y)$ is simple and has residue $\pm 1$. If $N_\pm(p)$ denotes the number of poles of $p$ of residue $\pm 1$, then
\begin{gather*}
\frac{1}{2\pi\ii}\oint\nolimits_Cp(y)\,\dd y = N_+(p)-N_-(p)\in\mathbb{Z},
\end{gather*}
so comparing with \eqref{eq:residue-integral}, we see that \eqref{eq:myPII} admits a rational solution $p(y)$ only if $m\in\mathbb{Z}$.

\subsubsection[B\"acklund transformations. Suf\/f\/icient conditions for existence of rational solutions. Uniqueness]{B\"acklund transformations. Suf\/f\/icient conditions for existence\\ of rational solutions. Uniqueness}

In 1971, Lukashevich \cite{Lukashevich71} discovered an explicit B\"acklund transformation for \eqref{eq:myPII}. Namely, supposing that $p(y)$ is an arbitrary solution of \eqref{eq:myPII}, the related function $\widehat{p}(y)$ def\/ined\footnote{The def\/inition apparently fails if $p(y)$ is such that $3p'(y)-3p(y)^2-y$ vanishes identically. It is easy to check that this condition is consistent with \eqref{eq:myPII} only if $m=-\tfrac{1}{2}$ (so the numerator in \eqref{eq:PII-Backlund} vanishes also). In the case $m=-\tfrac{1}{2}$, the Riccati equation $3p'(y)-3p(y)^2-y=0$ is solved by the formula $p(y)=-\phi'(y)/\phi(y)$ where $\phi''(y)+\tfrac{1}{3}y\phi(y)=0$; in other words, $\phi$ is an Airy function and $p(y)$ is a known special function (Airy) solution of~\eqref{eq:myPII}. A similar remark applies to the inverse transformation \eqref{eq:PII-Backlund-inverse} which gives $m=\tfrac{1}{2}$ and the denominator vanishes for the Airy solutions. Note that both transformations~\eqref{eq:PII-Backlund} and \eqref{eq:PII-Backlund-inverse} make sense whenever $m\in\mathbb{Z}$ except possibly for isolated values of~$y\in\mathbb{C}$.} by
\begin{gather}\label{eq:PII-Backlund}
\begin{split}
&\widehat{p}(y):=\frac{3p''(y)-3p(y)p'(y)-3p(y)^3-yp(y)-1}{3p'(y)-3p(y)^2-y}=-p(y)-\frac{2m+1}{3p'(y)-3p(y)^2-y},\\
& \widehat{m}:=m+1
\end{split}
\end{gather}
satisf\/ies \eqref{eq:myPII} with $(p,m)$ replaced by $(\widehat{p},\widehat{m})$. Note that $\widehat{p}$ is a rational function whenever $p$ is; therefore from the ``seed solution'' $p(y)=0$ for $m=0$, the (iterated) B\"acklund transformation produces a rational solution of \eqref{eq:myPII} for every positive integer $m$. From the elementary symmetry $(p(y),m)\mapsto (-p(y),-m)$ of \eqref{eq:myPII} one then has the existence of a rational solution of~\eqref{eq:myPII} for every $m\in\mathbb{Z}$, a fact that also follows from the earlier work of Yablonskii \cite{Yablonskii59} and Vorob'ev~\cite{Vorob'ev65} (see Section~\ref{sec:YVpolynomials} below). Hence, as pointed out by Airault \cite{Airault79}, the condition $m\in\mathbb{Z}$ is both necessary and suf\/f\/icient for the existence of rational solutions to~\eqref{eq:myPII}.

The B\"acklund transformation \eqref{eq:PII-Backlund} also yields a proof of uniqueness of the rational solution for given $m\in\mathbb{Z}$, a fact that was f\/irst noted by Murata \cite{Murata85}. Indeed, combining \eqref{eq:PII-Backlund} with the symmetry $(p(y),m)\mapsto (-p(y),-m)$ yields a second B\"acklund transformation
\begin{gather}\label{eq:PII-Backlund-inverse}
\begin{split}
& \widecheck{p}(y):=\frac{-3p''(y)-3p(y)p'(y)+3p(y)^3+yp(y)-1}{3p'(y)+3p(y)^2+y}=-p(y)+\frac{2m-1}{3p'(y)+3p(y)^2+y}\\
& \widecheck{m}:=m-1
\end{split}
\end{gather}
taking $p(y)$ solving \eqref{eq:myPII} to $\widecheck{p}(y)$ solving \eqref{eq:myPII} with $(p,m)$ replaced by $(\widecheck{p},\widecheck{m})$. Importantly, an elementary computation shows that the B\"acklund transformations~\eqref{eq:PII-Backlund} and~\eqref{eq:PII-Backlund-inverse} are inverses of each other on the space of solutions of~\eqref{eq:myPII}. In particular, both are injective maps. Suppose~$p(y)$ and~$\widetilde{p}(y)$ are distinct rational solutions of~\eqref{eq:myPII} for some $m\in\mathbb{Z}\setminus\{0\}$. Iteratively applying \mbox{either}~\eqref{eq:PII-Backlund} (if $m<0$) or~\eqref{eq:PII-Backlund-inverse} (if $m>0$) $|m|$ times, by injectivity we arrive at two distinct rational solutions of \eqref{eq:myPII} with $m=0$ in contradiction with the already-established uniqueness of the rational solution $p(y)=0$ for $m=0$. See~\cite{FokasGR93} for further information about B\"acklund transformations for Painlev\'e equations.

Explicitly, the f\/irst several rational Painlev\'e-II solutions are
\begin{gather}
p_0(y)=0,\qquad p_{\pm 1}(y)=\pm \frac{1}{y},\qquad p_{\pm 2}(y)=\pm \frac{2y^3-6}{y(y^3+6)},\nonumber\\
p_{\pm 3}(y)=\pm\frac{3y^2(y^6+12y^3+360)}{(y^3+6)(y^6+30y^3-180)},\label{eq:first-few-rationals}
\end{gather}
where the subscript indicates the value of $m$ in \eqref{eq:myPII}.

\subsubsection{Representation in terms of special polynomials}\label{sec:YVpolynomials}
It was f\/irst observed by Yablonskii \cite{Yablonskii59} and Vorob'ev \cite{Vorob'ev65} that rational solutions of the \mbox{Painle\-v\'e-II} equation \eqref{eq:myPII} can be represented in terms of special polynomials having an explicit recurrence relation. The \emph{Yablonskii--Vorob'ev polynomials} are def\/ined recursively by $Q_0(z):=1$, $Q_1(z):=z$, and then
\begin{gather}
Q_{n+1}(z):=\frac{zQ_n(z)^2-4\big(Q_n''(z)Q_n(z)-Q_n'(z)^2\big)}{Q_{n-1}(z)},\qquad n=1,2,3,\dots.
\label{eq:YV-recursion}
\end{gather}
Then the rational solution of \eqref{eq:myPII} can be expressed as
\begin{gather}
p(y)=p_m(y)=\frac{\dd}{\dd y}\ln\left(\frac{Q_m(z)}{Q_{m-1}(z)}\right),\qquad z=\left(\frac{2}{3}\right)^{1/3}y,\qquad m=1,2,3,\dots.\label{eq:p-YV}
\end{gather}

Of course the f\/irst surprise regarding the recursion formula \eqref{eq:YV-recursion} is that $\{Q_n(z)\}_{n=0}^\infty$ is indeed a sequence of polynomials in $z$. Indeed this is the case, as can be seen by the alternative formula due to Kajiwara and Ohta \cite{KajiwaraO96} expressing $Q_n(z)$ as a constant multiple of the Wronskian of $2n-1$ polynomials in~$z$ (see also \cite[Section~2.4]{Clarkson06}), and the f\/irst several iterates of~\eqref{eq:YV-recursion} produce:
\begin{gather*}
Q_2(z)=z^3+4,\qquad Q_3(z)=z^6+20z^3-80,\qquad Q_4(z)=z^{10}+60z^7+11200z.
\end{gather*}
It has been shown that the polynomials $\{Q_n(z)\}_{n=0}^\infty$ all have simple roots and that~$Q_m$ and~$Q_{m-1}$ can have no roots in common \cite{FukutaniOU00}, facts that are consistent via \eqref{eq:p-YV} with the fact that all poles of~$p(y)$ are simple with residues~$\pm 1$. Real roots of the Yablonskii--Vorob'ev polynomials correspond to real poles of~$p(y)$, and these have been studied extensively by Rof\/felsen who has shown that all nonzero real roots are all irrational~\cite{Roffelsen10} and that there are precisely $\lfloor(n+1)/3\rfloor$ negative roots of~$Q_n$ and $\lfloor(n+1)/2\rfloor$ total real roots of $Q_n$, and $Q_n(0)=0$ if and only if $n=1\pmod{3}$~\cite{Roffelsen12}. Also, the real roots of $Q_{n+1}$ and $Q_{n-1}$ interlace, as was proven by Clarkson~\cite{Clarkson06}.

\subsection{Outline of the paper}
The fact that all rational solutions of the Painlev\'e-II equation \eqref{eq:myPII} can be iteratively constructed, either via the direct B\"acklund transformations \eqref{eq:PII-Backlund} and \eqref{eq:PII-Backlund-inverse} or via the recurrence relation for the Yablonskii--Vorob'ev polynomials \eqref{eq:YV-recursion}, is quite remarkable and indicative of deeper integrable structure underlying the Painlev\'e-II equation. However, it must also be pointed out that the use of these iterative constructions is limited in practice, because the formulae generated become increasingly complicated as $|m|$ increases. The situation is similar to that encountered when studying orthogonal polynomials, which in general can be constructed systematically by a~Gram--Schmidt orthogonalization algorithm, but the number of steps of this algorithm increases with the degree of the polynomial desired, making it dif\/f\/icult to appeal to this approach to deduce properties of the general polynomial in the family.

Therefore, if our interest is to understand the analytic properties of the rational Painlev\'e-II functions, it is necessary to have an alternative representation that admits the possibility of asymptotic analysis for large $|m|$. In Section~\ref{sec:representations} we describe three such representations of the rational Painlev\'e-II solutions, two coming directly from the isomonodromic integrable structure underlying the Painlev\'e-II equation, and one related to a recently discovered representation of the squares of the Yablonskii--Vorob'ev polynomials in terms of the integrable structure behind orthogonal polynomials (which provides a work-around for the Gram--Schmidt procedure al\-lo\-wing large-degree asymptotics of general orthogonal polynomials to be computed). One of the contributions of our paper is then to establish a new identity relating the orthogonal polynomial approach to one of the isomonodromic approaches; see Section~\ref{sec:relation}.

These representations of the rational Painlev\'e-II solutions have indeed proven to be useful in characterizing the rational functions $p_m(y)$ in the limit of large $|m|$. In Section~\ref{sec:asymptotics} we review some of the results that have been proven with their help, outlining some of the methods of proof.

Below we will make frequent use of the Pauli spin matrices def\/ined by
\begin{gather*}
\sigma_1:=\begin{bmatrix}0 & 1\\1 & 0\end{bmatrix},\qquad\sigma_2:=\begin{bmatrix}0 & -\ii\\\ii & 0\end{bmatrix},\qquad\sigma_3:=\begin{bmatrix}1 & 0\\0 & -1\end{bmatrix}.
\end{gather*}

\section[Riemann--Hilbert problem representations of the rational Painlev\'e-II solution]{Riemann--Hilbert problem representations\\ of the rational Painlev\'e-II solutions}\label{sec:representations}
\subsection{Flaschka--Newell representation}\label{sec:FN}
In 1980, Flaschka and Newell \cite{FlaschkaN80} showed how a self-similar reduction of the Lax pair representation of the modif\/ied Korteweg--de Vries equation reveals the Painlev\'e-II equation in the form~\eqref{eq:myPII} to be an isomonodromic deformation of the linear equation
\begin{gather}
\frac{\partial\mathbf{v}}{\partial\lambda}=\mathbf{A}^\mathrm{FN}(\lambda,y)\mathbf{v},\qquad
\mathbf{A}^\mathrm{FN}(\lambda,y):=\begin{bmatrix}
-6\ii\lambda^2-3\ii p^2-\ii y & 6p\lambda + 3\ii p' + m\lambda^{-1}\\
6p\lambda-3\ii p'+m\lambda^{-1} & 6\ii \lambda^2+3\ii p^2+\ii y
\end{bmatrix}
\label{eq:PII-lambda}
\end{gather}
in which $p$, $p'$, $y$, and $m$ are regarded as numerical parameters. Indeed, \eqref{eq:PII-lambda} is compatible with the auxiliary linear equation
\begin{gather}
\frac{\partial\mathbf{v}}{\partial y}= \mathbf{B}^\mathrm{FN}(\lambda,y)\mathbf{v},\qquad\mathbf{B}^\mathrm{FN}(\lambda,y):=\begin{bmatrix} -\ii\lambda & p\\p & \ii\lambda
\end{bmatrix}\label{eq:PII-y}
\end{gather}
only if the
compatibility condition
\begin{gather}
\frac{\partial\mathbf{A}}{\partial y}-\frac{\partial\mathbf{B}}{\partial\lambda} + [\mathbf{A},\mathbf{B}]=\mathbf{0}\label{eq:ZCC}
\end{gather}
holds with $\mathbf{A}=\mathbf{A}^\mathrm{FN}$ and $\mathbf{B}=\mathbf{B}^\mathrm{FN}$. This forces $p$ to depend on $y$ by the Painlev\'e-II equation in the form \eqref{eq:myPII} and forces $p'=p'(y)$. The equation \eqref{eq:PII-y} then implies that the monodromy data associated with solutions of \eqref{eq:PII-lambda} depends trivially on $y$.

Let us describe the monodromy data associated with rational solutions $p=p_m(y)$ of \eqref{eq:myPII} for $m\in\mathbb{Z}$. It is pointed out in~\cite{FlaschkaN80} that whenever $(p,p')=(p_m(y),p'_m(y))$ in \eqref{eq:PII-lambda} for the rational solution $p_m(y)$, the irregular singular point at $\lambda=\infty$ for \eqref{eq:PII-lambda} exhibits only trivial Stokes phenomenon. This implies the existence of a
fundamental solution matrix of \eqref{eq:PII-lambda} of the form
\begin{gather}
\mathbf{V}_\infty(\lambda,y)=\left[\mathbb{I}+\sum_{n=1}^\infty \mathbf{K}^n(y)\lambda^{-n}\right]\ee^{-\ii\theta(\lambda,y)\sigma_3}
\label{eq:Vinfty}
\end{gather}
for some matrix coef\/f\/icients $\mathbf{K}^1(y)$, $\mathbf{K}^2(y)$, and so on, where
\begin{gather*}
\theta(\lambda,y):=2\lambda^3+y\lambda,
\end{gather*}
and where the inf\/inite series in \eqref{eq:Vinfty} is convergent for $|\lambda|$ suf\/f\/iciently large, which in view of $\lambda=0$ being the only f\/inite singular point actually means for $\lambda\neq 0$. Assuming compatibility, i.e., that $p=p(y)$ solves \eqref{eq:myPII} with $p'=p'(y)$, it can be shown that $\mathbf{V}_\infty(\lambda,y)$ is also a fundamental solution matrix for \eqref{eq:PII-y}, and then it follows by substitution into the latter system that $p(y)$ is recovered from the subleading term of the expansion \eqref{eq:Vinfty} by the formula
\begin{gather}
p(y)=2\ii K^1_{12}(y)=-2\ii K^1_{21}(y).\label{eq:p-from-Vinfty}
\end{gather}
On the other hand, $\lambda=0$ is a regular singular point for \eqref{eq:PII-lambda}.
Applying the method of Frobenius, there exists a fundamental solution matrix of \eqref{eq:PII-lambda} def\/ined in a neighborhood of
$\lambda=0$ having the form
\begin{gather}
\mathbf{V}_0(\lambda,y)=\left[\frac{1}{\sqrt{2}}\begin{bmatrix}1 & -1\\ 1 & 1\end{bmatrix}h(y)^{\sigma_3} + \sum_{n=1}^\infty\mathbf{H}^n(y)\lambda^n\right]\lambda^{m\sigma_3}
\label{eq:V0}
\end{gather}
for some scalar function $h(y)\neq 0$ and matrix coef\/f\/icients $\mathbf{H}^1(y)$, $\mathbf{H}^2(y)$, and so on. The absence of logarithms in spite of the fact that the Frobenius exponents $\pm m$ dif\/fer by an integer follows from the fact that, due to the triviality of the Stokes phenomenon at $\lambda=\infty$, the monodromy matrix for \eqref{eq:PII-lambda} corresponding to any loop about the origin is the identity, hence diagonalizable. However the same fact implies an ambiguity in the formula \eqref{eq:V0} in which the dominant column in the limit $\lambda\to 0$ is only determined up to addition of a multiple of the subdominant column. Flaschka and Newell \cite{FlaschkaN80} resolve this ambiguity as follows. They f\/irst observe that the subdominant column is well-def\/ined after the choice of the scalar $h(y)$, and from the recurrence relations determining the higher-order terms from the preceding terms a~predictable pattern emerges in which consecutive terms are alternating scalar multiples of the vectors $(1,1)^\top$ and $(-1,1)^\top$. A similar well-def\/ined alternating pattern holds for the dominant column, but only through the terms with $n\le 2|m|-1$, with the term for $n=2|m|$ satisfying an equation that is consistent but indeterminate. Here a choice is made: the term for $n=2|m|$ is taken to continue the alternating pattern of vectors $(1,1)^\top$ and $(-1,1)^\top$. Once this choice has been made, the alternating pattern again continues to all orders of the dominant column. In other words, Flaschka and Newell take $\mathbf{V}_0(\lambda,y)$ in the more specif\/ic form
\begin{gather}
\mathbf{V}_0(\lambda,y)=\frac{1}{\sqrt{2}}\begin{bmatrix}1 & -1\\1 & 1\end{bmatrix}\left(\sum_{n=0}^{\infty}\sigma_1^n\begin{bmatrix}h^n_{11}(y) & 0\\0 & h^n_{22}(y)\end{bmatrix}\lambda^n\right)\lambda^{m\sigma_3},\nonumber\\
 h^0_{11}(y)=h(y)=h^0_{22}(y)^{-1}.\label{eq:V0-again}
\end{gather}
There is then exactly one matrix solution of \eqref{eq:PII-lambda} of this form for a given scalar $h(y)$, and moreover, assuming compatibility, $h(y)$ can be chosen up to a constant scalar multiple so that $\mathbf{V}_0(\lambda,y)$ simultaneously solves~\eqref{eq:PII-y}. Again, the inf\/inite series appearing in~\eqref{eq:V0-again} is convergent near $\lambda=0$, and since there are no other f\/inite singular points it is actually convergent for all $\lambda\in\mathbb{C}$. By taking the limits $\lambda\to 0$ and $\lambda\to\infty$ respectively, Abel's theorem implies the identities $\det(\mathbf{V}_0(\lambda,y))=1$ and $\det(\mathbf{V}_\infty(\lambda,y))=1$ because the coef\/f\/icient matrix $\mathbf{A}^\mathrm{FN}(\lambda,y)$ in~\eqref{eq:PII-lambda} has zero trace. Therefore, as both $\mathbf{V}_\infty(\lambda,y)$ and (for a suitable choice of $h(y)$) $\mathbf{V}_0(\lambda,y)$ are simultaneous fundamental solution matrices for~\eqref{eq:PII-lambda} and \eqref{eq:PII-y} def\/ined in a common domain $0<|\lambda|<\infty$, there exists a constant unimodular matrix $\mathbf{G}_m$ such that
\begin{gather}
\mathbf{V}_{\infty}(\lambda,y)=\mathbf{V}_0(\lambda,y)\mathbf{G}_m,\qquad 0<|\lambda|<\infty.\label{eq:connection}
\end{gather}
The \emph{connection matrix} $\mathbf{G}_m$ is the monodromy data for the linear problem \eqref{eq:PII-lambda} in the case that $p=p_m(y)$ is a rational solution of \eqref{eq:myPII}. For more general solutions given $m\in\mathbb{Z}$, or for non-integral values of $m$, the monodromy data becomes augmented with six \emph{Stokes matrices} of alternating triangularity connecting solutions each having the form \eqref{eq:Vinfty} (but only as an asymptotic series, with no convergence properties implied) in six overlapping sectors of the irregular singular point at $\lambda=\infty$.

It is easy to see that $\overline{\mathbf{V}}(\lambda,y):=\sigma_1\mathbf{V}(-\lambda,y)\sigma_1$ is a fundamental solution matrix for the system \eqref{eq:PII-lambda} whenever $\mathbf{V}(\lambda,y)$ is. This substitution also leaves \eqref{eq:PII-y} invariant. Since $\mathbf{V}_\infty(\lambda,y)$ is uniquely determined from \eqref{eq:PII-lambda} and the leading term of its large-$\lambda$ asymptotic expansion (convergent in the trivial-monodromy case at hand for rational solutions $p=p_m(y)$),
we deduce the identity
\begin{gather}
\overline{\mathbf{V}}_\infty(\lambda,y)=\mathbf{V}_\infty(\lambda,y).\label{eq:V-infty-symmetry}
\end{gather}
Similarly, given the scalar $h(y)$, it follows from \eqref{eq:V0-again} that
\begin{gather}
\overline{\mathbf{V}}_0(\lambda,y)=\mathbf{V}_0(\lambda,y)\begin{bmatrix}0 & (-1)^m\\(-1)^{m+1} & 0\end{bmatrix}.\label{eq:V-zero-symmetry}
\end{gather}
Therefore, conjugating by $\sigma_1$ and replacing $\lambda\mapsto -\lambda$ in \eqref{eq:connection}, the use of the identities~\eqref{eq:V-infty-symmetry}--\eqref{eq:V-zero-symmetry} shows that also
\begin{gather*}
\mathbf{V}_\infty(\lambda,y)=\mathbf{V}_0(\lambda,y)(-1)^m\sigma_3\mathbf{G}_m\sigma_1,\qquad 0<|\lambda|<\infty,
\end{gather*}
and hence comparing again with \eqref{eq:connection} one sees that $\mathbf{G}_m=(-1)^m\sigma_3\mathbf{G}_m\sigma_1$. This matrix identity along with the condition that $\det(\mathbf{G}_m)=1$ implies that $\mathbf{G}_m$ necessarily has the form
\begin{gather}
\mathbf{G}_m=\begin{bmatrix}\alpha & (-1)^m\alpha\\
(-1)^{m+1}(2\alpha)^{-1} & (2\alpha)^{-1}\end{bmatrix}, \label{eq:G-form}
\end{gather}
where only the nonzero constant $\alpha$ is undetermined by symmetry.

We may now formulate a Riemann--Hilbert problem to recover $\mathbf{V}_\infty(\lambda,y)$ and $\mathbf{V}_0(\lambda,y)$, and hence also the rational Painlev\'e-II function $p_m(y)$, from the monodromy data, i.e., from the connection matrix $\mathbf{G}_m$. To this end, we def\/ine a matrix $\mathbf{M}^m(\lambda,y)$ by
\begin{gather*}
 \mathbf{M}^m(\lambda,y)= \begin{cases}
 \mathbf{V}_\infty(\lambda,y)\ee^{\ii\theta(\lambda,y))
 \sigma_3}\lambda^{-m\sigma_3}, &|\lambda|>1,\\
 \mathbf{V}_0(\lambda,y)\ee^{\ii\theta(\lambda,y)
 \sigma_3}\lambda^{-m\sigma_3}, &|\lambda|<1.
 \end{cases}
\end{gather*}
It is then clear that $\mathbf{M}^m(\lambda,y)$ solves the following Riemann--Hilbert problem.
\begin{rhp}[Flaschka--Newell representation]\label{rhp:FN} Let $m\in\mathbb{Z}$ and $y\in\mathbb{C}$ be given. Seek a~$2\times 2$ matrix-valued function $\mathbf{M}^m(\lambda,y)$ defined for $\lambda\in\mathbb{C}$, $|\lambda|\neq 1$, with the following properties:
\begin{itemize}
\item \textit{\textbf{Analyticity.}} $\mathbf{M}^m(\lambda,y)$ is analytic for $|\lambda|\neq 1$, taking continuous boundary values \linebreak $\mathbf{M}^m_+(\lambda,y)$ and $\mathbf{M}^m_-(\lambda,y)$ for $|\lambda|=1$ from the interior and exterior respectively of the unit circle.
\item \textit{\textbf{Jump condition.}} The boundary values are related by
\begin{gather*}
\mathbf{M}^m_+(\lambda,y)=\mathbf{M}^m_-(\lambda,y)\lambda^{m\sigma_3}\ee^{-\ii\theta(\lambda,y)\sigma_3}\mathbf{G}_m^{-1}
\ee^{\ii\theta(\lambda,y)\sigma_3}\lambda^{-m\sigma_3},\qquad |\lambda|=1.
\end{gather*}
\item \textit{\textbf{Normalization.}} The matrix $\mathbf{M}^m(\lambda,y)$ is normalized at $\lambda=\infty$ as follows:
\begin{gather*}
\lim_{\lambda\to\infty}\mathbf{M}^m(\lambda,y)\lambda^{m\sigma_3}=\mathbb{I},
\end{gather*}
where the limit may be taken in any direction.
\end{itemize}
\end{rhp}
The solution of this Riemann--Hilbert problem exists precisely for those values of $y\in\mathbb{C}$ that are not poles of $p_m(y)$. Given the solution $\mathbf{M}^m(\lambda,y)$, one extracts the rational Painlev\'e-II function $p_m(y)$ from the limit (cf.~\eqref{eq:p-from-Vinfty})
\begin{gather}
p_m(y)=2\ii\lim_{\lambda\to\infty}\lambda^{1+m}M^m_{12}(\lambda,y)=-2\ii\lim_{\lambda\to\infty}\lambda^{1-m}M^m_{21}(\lambda,y).\label{eq:pm-FN}
\end{gather}
Note also that without loss of generality one may take the constant $\alpha$ in \eqref{eq:G-form} to be $\alpha=1$, simply by re-def\/ining $\mathbf{M}^m(\lambda,y)$ within the unit circle by multiplication on the right by $\alpha^{\sigma_3}$. Such a re-def\/inition clearly does not af\/fect $\mathbf{M}^m(\lambda,y)$ for $|\lambda|>1$ and therefore has no essential ef\/fect on the reconstruction of $p_m(y)$.

Flaschka and Newell observe that Riemann--Hilbert Problem~\ref{rhp:FN} can be solved by reduction to f\/inite-dimensional linear algebra, resulting in determinantal formulae for $p_m(y)$ equivalent to iterated B\"acklund transformations studied by Airault \cite{Airault79}. To see this, note that uniqueness of solutions of Riemann--Hilbert Problem~\ref{rhp:FN} is an elementary consequence of Liouville's theorem, so it is suf\/f\/icient to construct a solution by any means. Now, $\mathbf{M}^m(\lambda,y)$ necessarily has a~convergent Laurent expansion about $\lambda=\infty$, suggesting to seek $\mathbf{M}^m(\lambda,y)$ as a suitable Laurent polynomial. In fact, assuming without loss of generality that $m\ge 0$, we may suppose that in the domain $|\lambda|>1$ the f\/irst row of $\mathbf{M}^m(\lambda,y)$ has the form
\begin{gather}
M^m_{11}(\lambda,y)=\lambda^{-m} + a_1(y)\lambda^{-m-1}+\cdots+a_{m-1}(y)\lambda^{1-2m} + a_m(y)\lambda^{-2m},\nonumber\\
M^m_{12}(\lambda,y)=b_1(y)\lambda^{m-1}+b_2(y)\lambda^{m-2}+\cdots + b_{m-1}(y)\lambda + b_m(y).\label{eq:ansatz}
\end{gather}
This ansatz clearly satisf\/ies the necessary analyticity condition for $|\lambda|>1$ as well as the normalization condition at $\lambda=\infty$. The jump condition can then be reinterpreted as requiring that the linear combinations
\begin{gather*}
M^m_{11+}(\lambda,y):=\frac{1}{2\alpha}\big[M^m_{11-}(\lambda,y) + (-1)^m\ee^{2\ii\theta(\lambda,y)}\lambda^{-2m}M^m_{12-}(\lambda,y)\big],\\
M^m_{12+}(\lambda,y):=\alpha\big[(-1)^{m+1}\ee^{-2\ii\theta(\lambda,y)}\lambda^{2m}M^m_{11-}(\lambda,y) + M^m_{12-}(\lambda,y)\big],
\end{gather*}
where the boundary values $M_{11-}^m(\lambda,y)$ and $M_{12-}^m(\lambda,y)$ are given by the ansatz \eqref{eq:ansatz}, both be analytic functions within the unit disk, where the only potential singularity is $\lambda=0$. The form of the ansatz automatically guarantees that this is the case for $M_{12+}^m(\lambda,y)$, but $M_{11+}^m(\lambda,y)$ has precisely $2m$ negative powers of~$\lambda$ whose coef\/f\/icients are required to vanish. It is easily seen that this amounts to a square inhomogeneous linear system of equations, explicit in terms of the Taylor coef\/f\/icients of $\ee^{\pm 2\ii\theta(\lambda,y)}$, on the $2m$ unknowns $a_1(y),\dots,a_m(y)$ and $b_1(y),\dots,b_m(y)$. The solution of this linear system by Cramer's rule gives the rational Painlev\'e-II function $p_m(y)$ in the form $p_m(y)=2\ii b_1(y)$. For example, in the case $m=2$ we require that
\begin{gather*}
M^2_{11+}(\lambda,y)=\lambda^{-2}+a_1(y)\lambda^{-3}+a_2(y)\lambda^{-4}+\ee^{2\ii\theta(\lambda,y)}\big(b_1(y)\lambda^{-3}+b_2(y)\lambda^{-4}\big)\\ \hphantom{M^2_{11+}(\lambda,y)}{} =(a_2(y)+b_2(y))\lambda^{-4} + (a_1(y)+b_1(y)+2\ii yb_2(y))\lambda^{-3} \\
\hphantom{M^2_{11+}(\lambda,y)=}{}+
\big(1+2\ii yb_1(y)-2y^2b_2(y)\big)\lambda^{-2}\! + \big({-}2y^2b_1(y)+4\ii\big(1\!-\!\tfrac{1}{3}y^3\big)b_2(y)\big)\lambda^{-1}\! + \mathcal{O}(1),
\end{gather*}
where the last term represents a function analytic at $\lambda=0$, be analytic at $\lambda=0$ from which one obtains $p_2(y)=2\ii b_1(y)=(2y^3-6)/(y(y^3+6))$ as expected (cf.~\eqref{eq:first-few-rationals}).

\subsection{Jimbo--Miwa representation}\label{sec:JM}
In 1981, Jimbo and Miwa \cite{JimboM81a} found a representation of the Painlev\'e-II equation as the compatibility condition for a Lax pair dif\/ferent from that found by Flaschka and Newell. We take Jimbo and Miwa's linear equations in the form
\begin{gather}
\frac{\partial\mathbf{v}}{\partial\zeta}=\mathbf{A}^\mathrm{JM}(\zeta,y)\mathbf{v},\qquad
\mathbf{A}^\mathrm{JM}(\zeta,y):=\begin{bmatrix}
-\tfrac{3}{2}\zeta^2 - 3\mathcal{U}\mathcal{V}-\tfrac{1}{2}y & 3\mathcal{U}\zeta + \mathcal{W}\\
-3\mathcal{V}\zeta -\mathcal{Z} & \tfrac{3}{2}\zeta^2 + 3\mathcal{U}\mathcal{V}+\tfrac{1}{2}y
\end{bmatrix}\label{eq:JM-zeta}
\end{gather}
and
\begin{gather}
\frac{\partial\mathbf{v}}{\partial y}=\mathbf{B}^\mathrm{JM}(\zeta,y)\mathbf{v},\qquad
\mathbf{B}^\mathrm{JM}(\zeta,y):=\begin{bmatrix}
-\tfrac{1}{2}\zeta & \mathcal{U}\\-\mathcal{V} & \tfrac{1}{2}\zeta
\end{bmatrix}\label{eq:JM-y}
\end{gather}
For this system, the compatibility condition \eqref{eq:ZCC} with $\mathbf{A}=\mathbf{A}^\mathrm{JM}$ and $\mathbf{B}=\mathbf{B}^\mathrm{JM}$ is equivalent to the following f\/irst-order system of equations:
\begin{gather}
\mathcal{U}'(y)=-\tfrac{1}{3}\mathcal{W}(y),\nonumber\\
\mathcal{V}'(y)=\tfrac{1}{3}\mathcal{Z}(y),\nonumber\\
\mathcal{W}'(y)=6\mathcal{U}(y)^2\mathcal{V}(y)+y\mathcal{U}(y),\nonumber\\
\mathcal{Z}'(y)=-6\mathcal{U}(y)\mathcal{V}(y)^2-y\mathcal{V}(y).\label{eq:PII-system}
\end{gather}
This system admits a f\/irst integral
\begin{gather}
m:=\mathcal{U}(y)\mathcal{Z}(y)+\mathcal{V}(y)\mathcal{W}(y)+\tfrac{1}{2}=\text{const},
\label{eq:first-integral}
\end{gather}
and then with $p(y)=\mathcal{U}'(y)/\mathcal{U}(y)$ the system \eqref{eq:PII-system} yields the Painlev\'e-II equation for $p(y)$ in the form \eqref{eq:myPII}.

As with the Flaschka--Newell approach, it is the problem~\eqref{eq:JM-zeta} whose analysis for f\/ixed $y$ determines the monodromy data, which is then independent of $y$ for simultaneous solutions of \eqref{eq:JM-zeta}--\eqref{eq:JM-y}. However, the direct monodromy problem \eqref{eq:JM-zeta} has a dif\/ferent character than in the Flaschka--Newell approach because \eqref{eq:JM-zeta} has only one singular point, an irregular singular point at inf\/inity, while \eqref{eq:PII-lambda} has in addition a regular singular point at the origin if $m\neq 0$. Thus, all solutions of \eqref{eq:JM-zeta} are entire functions of $\zeta$, and all monodromy data is generated only from the Stokes phenemonon about the singular point at inf\/inity. In particular, it is the case that for the rational solution $p=p_m(y)$ for $m\in\mathbb{Z}$, solutions of \eqref{eq:JM-zeta} exhibit nontrivial Stokes phenomenon in contrast to the situation in Flaschka--Newell theory.

The Stokes multipliers for \eqref{eq:JM-zeta} when $p=p_m(y)$ is the rational solution of \eqref{eq:myPII} for $m\in\mathbb{Z}$ can be inferred from the following Riemann--Hilbert problem, which arises naturally in the study of solutions of the sine-Gordon equation $\epsilon^2u_{tt}-\epsilon^2u_{xx}+\sin(u)=0$ in the semiclassical limit near certain critical points $(x,t)=(x_\mathrm{crit},0)$; see \cite[Section~5]{BuckinghamM12}.
\begin{rhp}[Jimbo--Miwa representation]\label{rhp:JM} Let $m\in\mathbb{Z}$ and $y\in\mathbb{C}$ be given. Seek a $2\times 2$ matrix-valued function $\mathbf{Z}^m(\zeta,y)$ be defined for $\zeta\in\mathbb{C}\setminus\Sigma$, where $\Sigma$ is the union of six rays $\Sigma:=\mathbb{R}\cup\ee^{\ii\pi/3}\mathbb{R}\cup\ee^{-\ii\pi/3}\mathbb{R}$, and having the following properties:
\begin{itemize}
\item\textit{\textbf{Analyticity.}} $\mathbf{Z}^m(\zeta,y)$ is analytic for $\zeta\in\mathbb{C}\setminus\Sigma$, taking continuous boundary values along the boundary of each component of this domain.
\item\textit{\textbf{Jump condition.}} Taking each ray of $\Sigma$ to be oriented in the direction away from the origin and given a point $\zeta$ on one of the rays using the notation $\mathbf{Z}_+^m(\zeta,y)$ $($resp.~$\mathbf{Z}_-^m(\zeta,y))$ to denote the boundary value taken at $\zeta\in\Sigma$ from the left $($resp.\ right$)$, the boundary values are related by
\begin{gather*}
\mathbf{Z}_+^m(\zeta,y)=\mathbf{Z}_-^m(\zeta,y)\ee^{-\phi(\zeta,y)\sigma_3}\mathbf{V}\ee^{\phi(\zeta,y)\sigma_3},\qquad \zeta\in\Sigma\setminus\{0\},\qquad\phi(\zeta,y):=\tfrac{1}{2}\zeta^3+\tfrac{1}{2}y\zeta,
\end{gather*}
where $\mathbf{V}$ is constant along each ray and is as shown in Fig.~{\rm \ref{fig:JM-RHP}}.
\item\textit{\textbf{Normalization.}} The matrix $\mathbf{Z}^m(\zeta,y)$ is normalized at $\zeta=\infty$ as follows:
\begin{gather*}
\lim_{\zeta\to\infty}\mathbf{Z}^m(\zeta,y)(-\zeta)^{(1-2m)\sigma_3/2}=\mathbb{I},
\end{gather*}
where the limit can be taken in any direction except the positive real axis, which is the branch cut for the principal branch of $(-\zeta)^{(1-2m)\sigma_3/2}$.
\end{itemize}
\end{rhp}
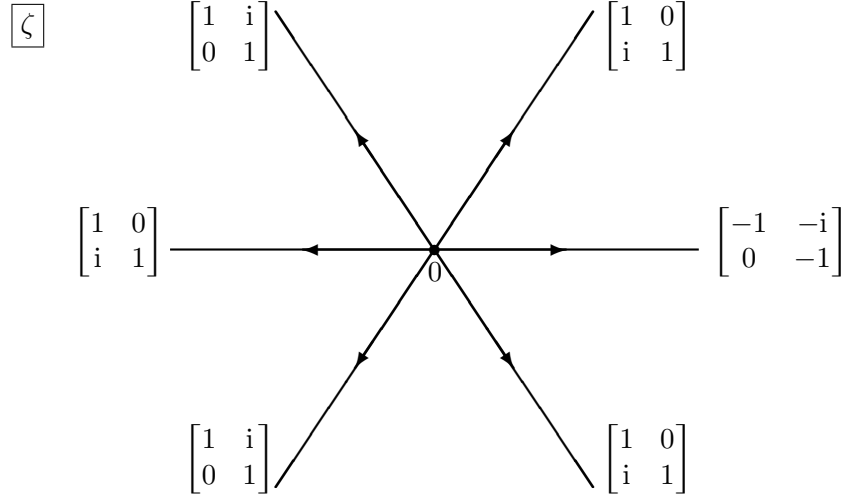
\begin{figure}[t]\centering
\setlength{\unitlength}{2pt}
\begin{picture}(100,100)(-50,-50)
\thicklines
\put(-80,41){\framebox{$\zeta$}}
\put(-50,0){\line(1,0){100}}
\put(0,0){\vector(-1,0){25}}
\put(0,0){\vector(1,0){25}}
\put(0,0){\line(2,-3){30}}
\put(0,0){\vector(2,-3){15}}
\put(0,0){\line(-2,3){30}}
\put(0,0){\vector(-2,3){15}}
\put(0,0){\line(2,3){30}}
\put(0,0){\vector(2,3){15}}
\put(0,0){\line(-2,-3){30}}
\put(0,0){\vector(-2,-3){15}}
\put(0,0){\circle*{2}}
\put(-1.25,-6){$0$}
\put(-68,0){$\begin{bmatrix} 1 & 0 \\ \ii & 1 \end{bmatrix}$}
\put(-47,39){$\begin{bmatrix} 1 & \ii \\ 0 & 1 \end{bmatrix}$}
\put(32,39){$\begin{bmatrix} 1 & 0 \\ \ii & 1 \end{bmatrix}$}
\put(53,0){$\begin{bmatrix} -1 & -\ii \\ 0 & -1 \end{bmatrix}$}
\put(32,-41){$\begin{bmatrix} 1 & 0 \\ \ii & 1 \end{bmatrix}$}
\put(-47,-41){$\begin{bmatrix} 1 & \ii \\ 0 & 1 \end{bmatrix}$}
\end{picture}
\caption{The jump contour $\Sigma$ and the value of the constant matrix $\mathbf{V}$ on each ray of $\Sigma$ for Riemann--Hilbert Problem~\ref{rhp:JM}.}\label{fig:JM-RHP}
\end{figure}

From the solution of Riemann--Hilbert Problem~\ref{rhp:JM} one obtains the rational Painlev\'e-II function $p_m(y)$ from the coef\/f\/icients in the large-$\zeta$ expansion of $\mathbf{Z}^m(\zeta,y)$:
\begin{gather}
\mathbf{Z}^m(\zeta,y)(-\zeta)^{(1-2m)\sigma_3/2}=\mathbb{I}+\mathbf{A}^m(y)\zeta^{-1}+
\mathbf{B}^m(y)\zeta^{-2}+\mathcal{O}\big(\zeta^{-3}\big),\qquad\zeta\to\infty,\label{eq:Z-expansion}
\end{gather}
by the formula
\begin{gather}
p_m(y)=A^m_{22}(y)-\frac{B^m_{12}(y)}{A^m_{12}(y)}.\label{eq:pm-from-JM}
\end{gather}
In \cite{BuckinghamM12}, it was deduced that Riemann--Hilbert Problem~\ref{rhp:JM} encodes the Stokes multipliers for the Lax pair \eqref{eq:JM-zeta}--\eqref{eq:JM-y} associated with the rational Painlev\'e-II function $p_m(y)$ as follows. Firstly, by considering $\mathbf{L}^m(\zeta,y):=\mathbf{Z}^m(\zeta,y)\ee^{-\phi(\zeta,y)\sigma_3}$, one shows that partial derivatives of $\mathbf{L}^m(\zeta,y)$ with respect to $\zeta$ and $y$ satisfy exactly the same jump conditions on the rays of $\Sigma$ as does $\mathbf{L}^m(\zeta,y)$ itself, a fact that along with some local analysis near $\zeta=0$ and $\zeta=\infty$ shows that $\mathbf{L}^m(\zeta,y)$ is a~simultaneous fundamental solution matrix of the two Lax pair equations \eqref{eq:JM-zeta}--\eqref{eq:JM-y}, provided that the coef\/f\/icients $\mathcal{U}$, $\mathcal{V}$, $\mathcal{W}$, and $\mathcal{Z}$ are def\/ined from the expansion \eqref{eq:Z-expansion}
by the formulae
\begin{gather*}
\mathcal{U}(y):=A^m_{12}(y),\qquad \mathcal{V}(y):=A^m_{21}(y),\qquad \mathcal{W}(y):=3B^m_{12}(y)-3A^m_{12}(y)A^m_{22}(y),\\
\mathcal{Z}(y):=3B^m_{21}(y)-3A^m_{21}(y)A^m_{11}(y).
\end{gather*}
Then, by reexamination of the asymptotic behavior of $\mathbf{L}^m(\zeta,y)$ for large $\zeta$ one f\/inds that the parameter $m\in\mathbb{Z}$ appearing in Riemann--Hilbert Problem~\ref{rhp:JM} is related to these functions by the identity \eqref{eq:first-integral}, identifying it with the parameter $m$ appearing in the Painlev\'e-II equation \eqref{eq:myPII}. It remains therefore to deduce that $p_m(y)$ def\/ined now by the expression \eqref{eq:pm-from-JM} is the \emph{rational} solution of \eqref{eq:myPII}. This can be accomplished by f\/irst noting that in the case $m=0$ a symmetry argument combined with~\eqref{eq:pm-from-JM} shows that $p_0(y)=0$, at which point one can leverage the $y$-part~\eqref{eq:JM-y} of the Lax pair to construct $\mathbf{Z}^0(\zeta,y)$ explicitly in terms of Airy functions of argument $6^{-1/3}\big(y+\tfrac{3}{2}\zeta^2\big)$. Then, one can apply iterated \emph{discrete isomondromic Schlesinger transformations} (also known in the integrable systems literature as \emph{Darboux transformations}; see \cite[Section~2]{BertolaC15} and \cite{JimboM81a} for further information on these notions) to explicitly increment or decrement the value of $m$ in integer steps, with the corresponding ef\/fect on the coef\/f\/icient $p_m(y)$ def\/ined by~\eqref{eq:pm-from-JM} being given by the B\"acklund transformations~\eqref{eq:PII-Backlund} or~\eqref{eq:PII-Backlund-inverse} respectively. As these preserve rationality, one concludes that $p_m(y)$ given by \eqref{eq:pm-from-JM} is precisely the rational solution of \eqref{eq:myPII} when $\mathbf{Z}^m(\zeta,y)$ is the solution of Riemann--Hilbert Problem~\ref{rhp:JM} for arbitrary $m\in\mathbb{Z}$. See \cite[Section~5]{BuckinghamM12} for full details of these arguments.

\subsection{Bertola--Bothner representation}
In \cite{BertolaB15}, Bertola and Bothner derived a new Hankel determinant representation of the squares of the Yablonskii--Vorob'ev polynomials $\{Q_n(z)\}_{n=0}^\infty$ def\/ined by the recurrence relation \eqref{eq:YV-recursion} with initial conditions $Q_0(z)=1$ and $Q_1(z)=z$. This new identity leads to a formula expressing the rational Painlev\'e-II function $p_m(y)$ in terms of pseudo-orthogonal polynomials (i.e., polynomials orthogonal with respect to an indef\/inite inner product involving contour integration with a~complex-valued weight), and this in turn leads to a Riemann--Hilbert representation.

The main theorem reported and proved in \cite{BertolaB15} is the following.
\begin{Theorem}[Bertola and Bothner, \cite{BertolaB15}]\label{thm:BB} Given $z\in\mathbb{C}$, let $\{\mu_k(z)\}_{k=0}^\infty$ denote the Taylor coefficients of the generating function $f(t):=\ee^{tz-\tfrac{1}{3}t^3}$:
\begin{gather*}
\ee^{tz-\tfrac{1}{3}t^3}=\sum_{k=0}^\infty\mu_k(z)t^k,\qquad (z,t)\in \mathbb{C}^2.
\end{gather*}
Then, for any $n\ge 1$,
\begin{gather*}
Q_{n-1}(z)^2 = (-1)^{\lfloor n/2\rfloor}\frac{{D}_n(z)}{2^{n-1}}\prod_{k=1}^{n-1}\left[\frac{(2k)!}{k!}\right]^2,
\end{gather*}
where $\lfloor u\rfloor$ denotes the greatest integer less than or equal to $u$ and ${D}_n(z)$ is the Hankel determinant
\begin{gather*}
{D}_n(z):=\det [\mu_{l+j-2}(z)]_{l,j=1}^n.
\end{gather*}
\end{Theorem}
The coef\/f\/icients $\mu_k(z)$ are polynomials with numerous special properties, some of which are enumerated in \cite{BertolaB15}. Similar determinantal representations of the Yablonskii--Vorob'ev polynomials themselves (not the squares) had been previously known \cite{KajiwaraO96}, including one represen\-ting~$Q_n(z)$ via a non-Hankel determinant involving the scaled functions $\mu_k\big(4^{1/3}z\big)$ and one represen\-ting~$Q_n(z)$ as a Hankel determinant built from functions that can be extracted from a~ge\-ne\-rating function via a non-convergent asymptotic series~\cite{IwasakiKN02}. However, it is the combination of the Hankel structure of the determinant with the convergent nature of the generating function expansion that leads to a~Riemann--Hilbert representation of $p_m(y)$ as we will now explain.

When combined with Theorem~\ref{thm:BB}, the representation \eqref{eq:p-YV} of $p_m(y)$ in terms of the Yab\-lons\-kii--Vorob'ev polynomials gives
\begin{gather}
p_m(y)=\frac{1}{2}\frac{\dd}{\dd y}\ln(\eta_m((\tfrac{2}{3})^{1/3}y)),\qquad\eta_m(z):=\frac{{D}_{m+1}(z)}{{D}_m(z)},\qquad m=1,2,3,\dots.
\label{eq:pm-norming-constant}
\end{gather}
Now, since the polynomials $\{\mu_k(z)\}_{k=0}^\infty$ are Taylor coef\/f\/icients of the entire function $f(t)=\ee^{tz-\tfrac{1}{3}t^3}$, they may be written as contour integrals using the Cauchy integral formula:
\begin{gather*}
\mu_k(z)=\left.\frac{1}{k!}\frac{\dd^k}{\dd t^k}\ee^{tz-\tfrac{1}{3}t^3}\right|_{t=0} = \frac{1}{2\pi\ii}\oint\nolimits_C t^{-k-1}\ee^{tz-\tfrac{1}{3}t^3}\,\dd t,\qquad k=0,1,2,3,\dots.
\end{gather*}
Here $C$ is a simple contour encircling the origin in the counterclockwise direction; without loss of generality we will take it to coincide with the unit circle. Setting $t=\xi^{-1}$ in the integrand puts the formula in the equivalent form
\begin{gather*}
\mu_k(z)=\oint\nolimits_C\xi^k\,\dd\nu(\xi;z),\qquad k=0,1,2,3,\dots,
\end{gather*}
where $C$ may be taken to be the same contour, and where
\begin{gather*}
\dd\nu(\xi;z):=\frac{\ee^{-\tfrac{1}{3}\xi^{-3}+\xi^{-1}z}}{2\pi\ii \xi}\,\dd \xi.
\end{gather*}
Thus, $\{\mu_k(z)\}_{k=0}^\infty$ are revealed as the monomial moments of a complex-valued weight pa\-ra\-met\-ri\-zed by $z\in\mathbb{C}$ and def\/ined on the unit circle. This fact immediately gives an interpretation to the ratio $\eta_m(z)$ of consecutive Hankel determinants (cf.~\eqref{eq:pm-norming-constant}); it is the \emph{norming constant} of the monic \emph{pseudo-orthogonal polynomial} $\psi_m(\xi;z)=\xi^m + c_{m,m-1}(z)\xi^{m-1}+\cdots + c_{m,1}(z)\xi + c_{m,0}(z)$ def\/ined given $z\in\mathbb{C}$ by the pseudo-orthogonality conditions
\begin{gather}
\oint\nolimits_C\psi_m(\xi;z)\xi^j\,\dd\nu(\xi;z)=0,\qquad j=0,1,2,\dots,m-1.\label{eq:orthogonality}
\end{gather}
Indeed, if $\psi_m(\xi;z)$ exists\footnote{Existence is not guaranteed for every $z\in\mathbb{C}$ because integration against $\dd\nu(\xi;z)$ does not def\/ine a def\/inite inner product, nor does \eqref{eq:orthogonality} represent Hermitian orthogonality which would require replacing $\xi^j$ with its complex conjugate. Hence the terminology of ``pseudo-orthogonality''.} for the given value of $z\in\mathbb{C}$ then it follows that
\begin{gather}
\eta_m(z)=\oint\nolimits_C\psi_m(\xi;z)\xi^m\,\dd\nu(\xi;z).\label{eq:norming-constant}
\end{gather}
The points $y\in\mathbb{C}$ where either $\psi_m\big(\xi;\big(\tfrac{2}{3}\big)^{1/3}y\big)$ fails to exist or $\eta_m\big(\big(\tfrac{2}{3}\big)^{1/3}y\big)=0$ (but possibly not both, should cancellation occur) are precisely the poles of $p_m(y)$.

Now, it is well-known that given any complex measure on a suitable contour, the corresponding pseudo-orthogonal polynomial of degree $m$ can be characterized via the solution of a~matrix Riemann--Hilbert problem of Fokas--Its--Kitaev type \cite{FokasIK91}. In the present context, that Riemann--Hilbert problem is the following.
\begin{rhp}[Bertola--Bothner representation]\label{rhp:BB}
Let $m\ge 0$ be an integer, and let $z\in\mathbb{C}$ be given. Seek a $2\times 2$ matrix-valued function $\mathbf{Y}^m(\xi,z)$ defined for $\xi\in\mathbb{C}$, $|\xi|\neq 1$, with the following properties:
\begin{itemize}
\item\textit{\textbf{Analyticity.}} $\mathbf{Y}^m(\xi,z)$ is analytic for $|\xi|\neq 1$, taking continuous boundary values \linebreak $\mathbf{Y}^m_+(\xi,z)$ and $\mathbf{Y}^m_-(\xi,z)$ for $|\xi|=1$ from the interior and exterior respectively of the unit circle.
\item\textit{\textbf{Jump condition.}} The boundary values are related by
\begin{gather}
\mathbf{Y}^m_+(\xi,z)=\mathbf{Y}^m_-(\xi,z)\begin{bmatrix}1 & \nu'(\xi;z)\\0 & 1\end{bmatrix},\qquad |\xi|=1,\nonumber\\
\nu'(\xi;z):=\frac{\dd\nu(\xi;z)}{\dd \xi}=\frac{\ee^{-\tfrac{1}{3}\xi^{-3}+z\xi^{-1}}}{2\pi\ii \xi}.\label{eq:BB-jump}
\end{gather}
\item\textit{\textbf{Normalization.}} The matrix $\mathbf{Y}^m(\xi,z)$ is normalized at $\xi=\infty$ as follows:
\begin{gather*}
\lim_{\xi\to\infty}\mathbf{Y}^m(\xi,z)\xi^{-m\sigma_3}=\mathbb{I},
\end{gather*}
where the limit may be taken in any direction.
\end{itemize}
\end{rhp}
Indeed, all of the relevant quantities associated with the pseudo-orthogonal polynomials for the weight $\dd\nu(\xi;z)$ are encoded in the solution of this problem. In particular,
\begin{gather*}
Y^m_{11}(\xi,z)=\psi_m(\xi;z)\qquad\text{and}\qquad Y^m_{12}(\xi,z)=\frac{1}{2\pi\ii}\oint\nolimits_C\frac{\psi_m(w;z)\,\dd\nu(w;z)}{w-\xi},
\end{gather*}
from which it follows (cf.\ \eqref{eq:orthogonality}--\eqref{eq:norming-constant}) that
\begin{gather*}
\eta_m(z)=-2\pi\ii\lim_{\xi\to\infty}\xi^{m+1}Y_{12}^m(\xi,z).
\end{gather*}

Asymptotic analysis of the pseudo-orthogonal polynomials $\psi_m(\xi;z)$ in the limit of large $m$ can therefore be carried out by applying steepest descent techniques to Riemann--Hilbert Problem~\ref{rhp:BB}, as was f\/irst done in the case of true orthogonality on the real line in~\cite{DeiftKMVZ99} and in the case of true orthogonality on the unit circle in \cite{BaikDJ99}. However, noting that the expression~\eqref{eq:pm-norming-constant} involves dif\/ferentiation with respect to the parameter $z$, a limit process that cannot be assumed to commute with the limit $m\to\infty$, Bertola and Bothner show how to obtain the relevant derivatives directly from the solution $\mathbf{Y}^m(\xi,z)$ of Riemann--Hilbert Problem~\ref{rhp:BB}. The essence of the argument is as follows. The related matrix $\mathbf{N}^m(\xi,z):= \mathbf{Y}^m(\xi,z)\ee^{z\xi^{-1}\sigma_3/2}$ must be analytic for $\xi\in\mathbb{C}\setminus\{0\}$ and satisf\/ies jump condition across the unit circle of exactly the form \eqref{eq:BB-jump} in which $z$ has been replaced by $z=0$. As the parameter $z$ no longer appears in the jump mat\-rix for $\mathbf{N}^m(\xi,z)$, it follows that the partial derivative $\mathbf{N}_z^m(\xi,z)$ also satisf\/ies exactly the same jump condition, and therefore the matrix ratio $\mathbf{N}_z^m(\xi,z)\mathbf{N}^m(\xi,z)^{-1}$ has no jump and so extends to an analytic function on the punctured complex plane $\mathbb{C}\setminus\{0\}$. The asymptotic behavior of $\mathbf{N}_z^m(\xi,z)\mathbf{N}^m(\xi,z)^{-1}$ for large and small $\xi$ is easily expressed in terms of $\mathbf{Y}^m(\xi,z)$:
\begin{gather*}
\mathbf{N}_z^m(\xi,z)\mathbf{N}^m(\xi,z)^{-1} = \begin{cases}
\big(\mathbf{Y}^{m\prime}_1(z)+\frac{1}{2}\sigma_3\big)\xi^{-1}+\mathcal{O}\big(\xi^{-2}\big),& \xi\to\infty,\\
\frac{1}{2}\mathbf{Y}^m(0,z)\sigma_3\mathbf{Y}^m(0,z)^{-1}\xi^{-1}+\mathcal{O}(1),& \xi\to 0,
\end{cases}
\end{gather*}
where $\mathbf{Y}^m(\xi,z)\xi^{-m\sigma_3}=\mathbb{I}+\mathbf{Y}_1^m(z)\xi^{-1}+\mathcal{O}(\xi^{-2})$ as $\xi\to\infty$. Therefore $\mathbf{N}_z^m(\xi,z)\mathbf{N}^m(\xi,z)^{-1}$ is a~$z$-dependent multiple of $\xi^{-1}$ given by two equivalent formulae:
\begin{gather*}
\mathbf{N}_z^m(\xi,z)\mathbf{N}^m(\xi,z)^{-1}=\big(\mathbf{Y}^{m\prime}_1(z)+\tfrac{1}{2}\sigma_3\big)\xi^{-1}=
\tfrac{1}{2}\mathbf{Y}^m(0,z)\sigma_3\mathbf{Y}^m(0,z)^{-1}\xi^{-1}.
\end{gather*}
From the $(1,2)$-entry in this matrix identity one obtains
\begin{gather*}
\eta_m'(z)=-2\pi\ii Y^{m\prime}_{1,12}(z)=2\pi\ii Y_{11}^m(0,z)Y_{21}^m(0,z),\qquad m=0,1,2,\dots,
\end{gather*}
where we have used the fact that the necessarily unique solution of Riemann--Hilbert Problem~\ref{rhp:BB} has unit determinant. Therefore, from the solution of Riemann--Hilbert Problem~\ref{rhp:BB} the rational Painlev\'e-II function $p_m(y)$ can be expressed without dif\/ferentiation with respect to~$z$ as
\begin{gather}
p_m(y)=-\frac{Y_{11}^m(0,z)Y_{12}^m(0,z)}{12^{1/3}Y_{1,12}^m(z)},\qquad z=\left(\frac{2}{3}\right)^{1/3}y,\nonumber\\ \mathbf{Y}^m_1(z):=\lim_{\xi\to\infty}\xi\big(\mathbf{Y}^m(\xi,z)\xi^{-m\sigma_3}-\mathbb{I}\big),\label{eq:pm-BB}
\end{gather}
for $m=0,1,2,\dots$.

\subsection[Explicit relation between the Flaschka--Newell and Bertola--Bothner representations]{Explicit relation between the Flaschka--Newell\\ and Bertola--Bothner representations}\label{sec:relation}
The Riemann--Hilbert representations of the rational Painlev\'e-II functions appearing in the isomonodromy approaches of Flaschka--Newell (cf.\ Section~\ref{sec:FN}) and Jimbo--Miwa (cf.\ Section~\ref{sec:JM}) are known to be related. Indeed, Joshi, Kitaev, and Treharne found an explicit integral transform relating simultaneous solutions of the corresponding Lax pairs \cite[Corollary~3.2]{JoshiKT09}. This integral transform provides another explanation for the fact that the solution of Riemann--Hilbert Problem~\ref{rhp:FN} is rational in $\lambda$ while that of Riemann--Hilbert Problem~\ref{rhp:JM} is transcendental in $\zeta$, being built from Airy functions \cite{BuckinghamM12}. The approach of Bertola--Bothner also leads to a Riemann--Hilbert representation of the rational Painlev\'e-II functions, but the approach is not motivated by isomonodromy theory for any Lax pair, so it seems more mysterious from this point of view. In this section we show that the Riemann--Hilbert problem appearing in the Bertola--Bothner approach is in fact explicitly connected to that arising in the Flaschka--Newell isomonodromy theory:
\begin{Theorem}\label{thm:connection}Let $m\ge 0$ be an integer, suppose that $y\in\mathbb{C}$ is not a pole of the rational \mbox{Painlev\'e-II} function $p_m(y)$, and let $z=\big(\tfrac{2}{3}\big)^{1/3}y$. Then the unique solution $\mathbf{M}^m(\lambda,y)$ of Riemann--Hilbert Problem~{\rm \ref{rhp:FN}} arising from Flaschka--Newell theory is related to the unique solution $\mathbf{Y}^m(\xi,z)$ of Riemann--Hilbert Problem~{\rm \ref{rhp:BB}} arising from the Bertola--Bothner approach by an explicit elementary transformation with an explicit elementary inverse $($cf.\ equations~\eqref{eq:A-def}--\eqref{eq:B-def}, \eqref{eq:YB}, \eqref{eq:C-expansion}, \eqref{eq:beta-down}, and~\eqref{eq:beta-up} in the proof below$)$.
\end{Theorem}
\begin{proof}We start with the Flaschka--Newell approach and Riemann--Hilbert Problem~\ref{rhp:FN}. Suppose without loss of generality that $m=1,2,3,\dots$. We begin by noting that the matrix $\mathbf{G}_m^{-1}$ def\/ined by~\eqref{eq:G-form} has the lower-upper factorization
\begin{gather*}
\mathbf{G}_m^{-1}=\begin{bmatrix}(2\alpha)^{-1}& (-1)^{m+1}\alpha\\ (-1)^m(2\alpha)^{-1} & \alpha\end{bmatrix}=
\begin{bmatrix} 1 & 0\\(-1)^{m} & 1\end{bmatrix}\begin{bmatrix}(2\alpha)^{-1} & (-1)^{m+1}\alpha \\ 0 & 2\alpha\end{bmatrix},
\end{gather*}
and therefore the jump matrix in Riemann--Hilbert Problem~\ref{rhp:FN} is
\begin{gather*}
\lambda^{m\sigma_3}\ee^{-\ii\theta(\lambda,y)\sigma_3}\mathbf{G}_m^{-1}\ee^{\ii\theta(\lambda,y)\sigma_3}\lambda^{-m\sigma_3}\\
\qquad {}= \begin{bmatrix}1 & 0\\(-1)^m\lambda^{-2m}\ee^{2\ii\theta(\lambda,y)} & 1\end{bmatrix}
\begin{bmatrix}(2\alpha)^{-1} & (-1)^{m+1}\alpha\lambda^{2m}\ee^{-2\ii\theta(\lambda,y)}\\
0 & 2\alpha\end{bmatrix},
\end{gather*}
and the right-hand factor is obviously analytic within the unit disk and has unit determinant. Therefore, def\/ining a new matrix $\mathbf{P}^m(\lambda,y)$ in terms of the unknown $\mathbf{M}^m(\lambda,y)$ by
\begin{gather}
\mathbf{P}^m(\lambda,y):=\begin{cases}\mathbf{M}^m(\lambda,y),& |\lambda|>1,\\
\mathbf{M}^m(\lambda,y)\begin{bmatrix}(2\alpha)^{-1} & (-1)^{m+1}\alpha\lambda^{2m}\ee^{-2\ii\theta(\lambda,y)}\\ 0 & 2\alpha\end{bmatrix}^{-1},& |\lambda|<1,
\end{cases}\label{eq:A-def}
\end{gather}
we see that $\mathbf{P}^m(\lambda,y)$ satisf\/ies exactly the same conditions as specif\/ied in Riemann--Hilbert Problem~\ref{rhp:FN} except that the jump condition across the unit circle becomes instead
\begin{gather}
\mathbf{P}^m_+(\lambda,y)=\mathbf{P}^m_-(\lambda,y)\begin{bmatrix}1 & 0\\(-1)^m\lambda^{-2m}\ee^{2\ii\theta(\lambda,y)} & 1\end{bmatrix},\qquad |\lambda|=1.
\end{gather}
This triangular jump matrix already suggests the Fokas--Its--Kitaev form that appears in the approach of Bertola and Bothner, but we require two more steps to complete the identif\/ication. Firstly, we make the simple substitution
\begin{gather}
\mathbf{Q}^m(\xi,z):=k^{\sigma_3}\sigma_1\mathbf{P}^m\big(0,\big(\tfrac{3}{2}\big)^{1/3}z\big)^{-1}\mathbf{P}^m\big(c\xi^{-1},
\big(\tfrac{3}{2}\big)^{1/3}z\big)\xi^{-m\sigma_3}\sigma_1k^{-\sigma_3},\label{eq:B-def}
\end{gather}
where
\begin{gather*}
c:=-\ii \cdot 12^{-1/3}\qquad \text{and}\qquad k:=\frac{\ii^{m+1}c^m}{\ee^{\ii\pi/4}\sqrt{2\pi}}.
\end{gather*}
Now observe that the following Riemann--Hilbert problem captures at the same time the matrix $\mathbf{Q}^m(\xi,z)$ and the matrix $\mathbf{Y}^m(\xi,z)$ appearing in the Bertola--Bothner approach, for dif\/ferent values of the auxiliary parameter $j\in\mathbb{Z}$.
\begin{rhp}\label{rhp:intermediate-1}
Let $m\in\mathbb{Z}$, $j\in\mathbb{Z}$, and $z\in\mathbb{C}$ be given. Seek a $2\times 2$ matrix-valued function $\mathbf{C}^{m,j}(\xi,z)$ defined for $\xi\in\mathbb{C}$, $|\xi|\neq 1$, with the following properties:
\begin{itemize}
\item\textit{\textbf{Analyticity.}} $\mathbf{C}^{m,j}(\xi,z)$ is analytic for $|\xi|\neq 1$, taking continuous boundary values\linebreak $\mathbf{C}^{m,j}_+(\xi,z)$ and $\mathbf{C}^{m,j}_-(\xi,z)$ for $|\xi|=1$ from the interior and exterior respectively of the unit circle.
\item\textit{\textbf{Jump condition.}} The boundary values are related by
\begin{gather}
\mathbf{C}^{m,j}_+(\xi,z)=\mathbf{C}^{m,j}_-(\xi,z)\begin{bmatrix}1 & \xi^j\nu'(\xi;z)\\0 & 1\end{bmatrix},\qquad |\xi|=1,\nonumber\\ \nu'(\xi;z)=\frac{\ee^{-\tfrac{1}{3}\xi^{-3}+z\xi^{-1}}}{2\pi\ii \xi}.\label{eq:intermediate-1-jump}
\end{gather}
\item\textit{\textbf{Normalization.}} The matrix $\mathbf{C}^{m,j}(\xi,z)$ is normalized at $\xi=\infty$ as follows:
\begin{gather*}
\lim_{\xi\to\infty}\mathbf{C}^{m,j}(\xi,z)\xi^{-m\sigma_3}=\mathbb{I},
\end{gather*}
where the limit may be taken in any direction.
\end{itemize}
\end{rhp}
Indeed, it is easy to check that
\begin{gather}
\mathbf{Q}^m(\xi,z)=\mathbf{C}^{m,1}(\xi,z)\qquad\text{and, for $m\ge 0$,}\qquad \mathbf{Y}^m(\xi,z)=\mathbf{C}^{m,0}(\xi,z) \label{eq:YB}
\end{gather}
by comparison with the conditions of Riemann--Hilbert Problems~\ref{rhp:FN} and~\ref{rhp:BB}. We complete the connection between the Flaschka--Newell and Bertola--Bothner approaches by next establishing the relation between solutions~$\mathbf{C}^{m,j}(\xi,z)$ for consecutive values of $j\in\mathbb{Z}$.

The solution $\mathbf{C}^{m,j}(\xi,z)$ of Riemann--Hilbert Problem~\ref{rhp:intermediate-1} has a convergent Laurent expansion for large $|\xi|$ of the form
\begin{gather}
\mathbf{C}^{m,j}(\xi,z)=\big(\mathbb{I} + \mathbf{R}^{m,j}(z)\xi^{-1}+\mathcal{O}\big(\xi^{-2}\big)\big)\xi^{m\sigma_3},\qquad \xi\to\infty \label{eq:C-expansion}
\end{gather}
for some residue matrix $\mathbf{R}^{m,j}(z)$. Noting that if it exists for a given $z\in\mathbb{C}$, the unique solution of Riemann--Hilbert Problem~\ref{rhp:intermediate-1} has unit determinant, consider the matrix $\widecheck{\mathbf{E}}(\xi,z)$ def\/ined by
\begin{gather}
\widecheck{\mathbf{E}}(\xi,z):=\mathbf{C}^{m,j}(\xi,z)\begin{bmatrix}1 & 0\\0 & \xi\end{bmatrix}
\mathbf{C}^{m,j+1}(\xi,z)^{-1},\qquad |\xi|\neq 1. \label{eq:E-define}
\end{gather}
It is straightforward to check from \eqref{eq:intermediate-1-jump} that the boundary values taken by $\widecheck{\mathbf{E}}(\xi,z)$ on the unit circle satisfy the trivial jump condition $\widecheck{\mathbf{E}}_+(\xi,z)=\widecheck{\mathbf{E}}_-(\xi,z)$ for $|\xi|=1$; hence $\widecheck{\mathbf{E}}(\xi,z)$ extends to the whole complex plane as an entire function of $\xi$. Moreover, using~\eqref{eq:C-expansion} it follows that~$\widecheck{\mathbf{E}}(\xi,z)$ has the following asymptotic expansion for large~$\xi$:
\begin{gather}
\widecheck{\mathbf{E}}(\xi,z)=\big(\mathbb{I}+\mathbf{R}^{m,j}(z)\xi^{-1}+\mathcal{O}\big(\xi^{-2}\big)\big)
\begin{bmatrix}1&0\\0 & \xi\end{bmatrix}\big(\mathbb{I}-\mathbf{R}^{m,j+1}(z)\xi^{-1}+\mathcal{O}\big(\xi^{-2}\big)\big)\nonumber\\
\hphantom{\widecheck{\mathbf{E}}(\xi,z)}{} =\begin{bmatrix}
1 & R_{12}^{m,j}(z)\\
-R_{21}^{m,j+1}(z) & \xi +R_{22}^{m,j}(z)-R_{22}^{m,j+1}(z)
\end{bmatrix} + \mathcal{O}\big(\xi^{-1}\big),\qquad \xi\to\infty. \label{eq:Laurent-E}
\end{gather}
It then follows by Liouville's theorem that all negative power terms in the Laurent expansion of~$\widecheck{\mathbf{E}}(\xi,z)$ vanish, i.e., $\widecheck{\mathbf{E}}(\xi,z)$ is the linear function of $\xi$ given by the explicit matrix on the second line of \eqref{eq:Laurent-E}. Returning to \eqref{eq:E-define}, we have established the identity
\begin{gather}
\mathbf{C}^{m,j}(\xi,z)\begin{bmatrix}1 & 0\\0 & \xi\end{bmatrix}=\begin{bmatrix}
1 & R_{12}^{m,j}(z)\\
-R_{21}^{m,j+1}(z) & \xi +R_{22}^{m,j}(z)-R_{22}^{m,j+1}(z)
\end{bmatrix}\mathbf{C}^{m,j+1}(\xi,z),\nonumber\\
|\xi|\neq 1.\label{eq:C-identity}
\end{gather}
If we can express the second column of $\mathbf{R}^{m,j}(z)$ in terms of $\mathbf{C}^{m,j+1}(\xi,z)$, then this becomes an explicit formula for $\mathbf{C}^{m,j}(\xi,z)$ in terms of the latter.

To this end, consider the second column of \eqref{eq:C-identity} evaluated at $\xi=0$, which reads
\begin{gather*}
\begin{bmatrix}0\\0\end{bmatrix}=\begin{bmatrix}C^{m,j+1}_{12}(0,z)+R_{12}^{m,j}(z)C^{m,j+1}_{22}(0,z)\\
-R_{21}^{m,j+1}(z)C^{m,j+1}_{12}(0,z)+(R_{22}^{m,j}(z)-R_{22}^{m,j+1}(z))C^{m,j+1}_{22}(0,z)\end{bmatrix}
\end{gather*}
because $\mathbf{C}^{m,j}(\xi,z)$ and $\mathbf{C}^{m,j+1}(\xi,z)$ are analytic at $z=0$. Therefore,
\begin{gather*}
R_{12}^{m,j}(z)=-\frac{C_{12}^{m,j+1}(0,z)}{C_{22}^{m,j+1}(0,z)}\qquad\text{and}\qquad
R_{22}^{m,j}(z)=R_{22}^{m,j+1}(z)+\frac{C_{12}^{m,j+1}(0,z)}{C_{22}^{m,j+1}(0,z)}R_{21}^{m,j+1}(z),
\end{gather*}
so substituting into \eqref{eq:C-identity} we recover the explicit formula for decrementing the value of $j$:
\begin{gather}
\mathbf{C}^{m,j}(\xi,z)=\widecheck{\mathbf{E}}(\xi,z)\mathbf{C}^{m,j+1}(\xi,z)\begin{bmatrix}1 & 0\\ 0 & \xi^{-1}\end{bmatrix},
\qquad\text{where}\nonumber\\
\widecheck{\mathbf{E}}(\xi,z)=\begin{bmatrix}1 & -C_{12}^{m,j+1}(0,z)C_{22}^{m,j+1}(0,z)^{-1}\\
-R_{21}^{m,j+1}(z) & \xi +R_{21}^{m,j+1}(z)C_{12}^{m,j+1}(0,z)C_{22}^{m,j+1}(0,z)^{-1}\end{bmatrix}.\label{eq:beta-down}
\end{gather}

In a similar way, the matrix
\begin{gather*}
\widehat{\mathbf{E}}(\xi,z):=\mathbf{C}^{m,j+1}(\xi,z)\begin{bmatrix}\xi & 0\\0 & 1\end{bmatrix}\mathbf{C}^{m,j}(\xi,z)^{-1}
\end{gather*}
is an entire function that equals the polynomial part of its Laurent expansion for large $\xi$, and hence
\begin{gather*}
\widehat{\mathbf{E}}(\xi,z)=\begin{bmatrix}
\xi+R_{11}^{m,j+1}(z)-R_{11}^{m,j}(z) & -R_{12}^{m,j}(z)\\
R_{12}^{m,j+1}(z) & 1
\end{bmatrix},
\end{gather*}
leading to the following analogue of \eqref{eq:C-identity}:
\begin{gather}
\mathbf{C}^{m,j+1}(\xi,z)\begin{bmatrix}\xi & 0\\0 & 1\end{bmatrix}=\begin{bmatrix}
\xi+R_{11}^{m,j+1}(z)-R_{11}^{m,j}(z) & -R_{12}^{m,j}(z)\\
R_{12}^{m,j+1}(z) & 1
\end{bmatrix}\mathbf{C}^{m,j}(\xi,z).\label{eq:C-identity-2}
\end{gather}
From the f\/irst column of \eqref{eq:C-identity-2} evaluated at $\xi=0$ we get
\begin{gather*}
R_{12}^{m,j+1}(z)=-\frac{C_{21}^{m,j}(0,z)}{C_{11}^{m,j}(0,z)}\qquad\text{and}\qquad
R_{11}^{m,j+1}(z)=R_{11}^{m,j}(z)+\frac{C_{21}^{m,j}(0,z)}{C_{11}^{m,j}(0,z)}R_{12}^{m,j}(z),
\end{gather*}
so substituting into \eqref{eq:C-identity-2} we recover the explicit formula for incrementing the value of $j$:
\begin{gather}
\mathbf{C}^{m,j+1}(\xi,z)=\widehat{\mathbf{E}}(\xi,z)\mathbf{C}^{m,j}(\xi,z)\begin{bmatrix}\xi^{-1}&0\\0 & 1\end{bmatrix},\qquad\text{where}\nonumber\\
\widehat{\mathbf{E}}(\xi,z)=\begin{bmatrix}
\xi+R_{12}^{m,j}(z)C_{21}^{m,j}(0,z)C_{11}^{m,j}(0,z)^{-1} & -R_{12}^{m,j}(z)\\
-C_{21}^{m,j}(0,z)C_{11}^{m,j}(0,z)^{-1} & 1
\end{bmatrix}.\label{eq:beta-up}
\end{gather}
Note that equations \eqref{eq:beta-down} and \eqref{eq:beta-up} can be interpreted as discrete Schlesinger/Darboux transformations (see \cite[Section 2]{BertolaC15} and \cite{JimboM81a}) for Riemann--Hilbert Problem~\ref{rhp:intermediate-1}.

Taking into account the explicit and obviously invertible transformations \eqref{eq:A-def}--\eqref{eq:B-def} re\-la\-ting $\mathbf{M}^{m}(\lambda,y)$ solving Riemann--Hilbert Problem~\ref{rhp:FN} to $\mathbf{Q}^{m}(\xi,z)=\mathbf{C}^{m,1}(\xi,z)$ via $\mathbf{P}^m(\lambda,y)$, the formulae \eqref{eq:beta-down} and \eqref{eq:beta-up} establish the connection with Riemann--Hilbert Problem~\ref{rhp:BB} having solution $\mathbf{Y}^m(\xi,z)=\mathbf{C}^{m,0}(\xi,z)$.
\end{proof}

We remark that although Theorem~\ref{thm:connection} provides an explicit relation between the solutions of Riemann--Hilbert Problems~\ref{rhp:FN} and \ref{rhp:BB}, it can happen that for given $z\in\mathbb{C}$ one of these problems is solvable and the other is not. This occurs precisely when one of the denominators $C_{22}^{m,1}(0,z)$ in \eqref{eq:beta-down} or $C_{11}^{m,0}(0,z)$ in~\eqref{eq:beta-up} vanishes. Indeed, we have mentioned before (and it actually follows from the formula~\eqref{eq:pm-FN}) that the points $z$ where Riemann--Hilbert Problem~\ref{rhp:FN} fails to be solvable correspond precisely to the poles of $p_m$. On the other hand, the formula \eqref{eq:pm-BB} shows that it is possible that some poles of $p_m$ can arise from the well-def\/ined function $Y^m_{1,12}$ vanishing at a point $z$ where Riemann--Hilbert Problem~\ref{rhp:BB} has a solution; hence Riemann--Hilbert Problem~\ref{rhp:BB} is solvable while Riemann--Hilbert Problem~\ref{rhp:FN} is not. It can also happen that Riemann--Hilbert Problem~\ref{rhp:BB} fails to be solvable at a point $z$ corresponding to a regular point of $p_m$ and hence a point of solvability of Riemann--Hilbert Problem~\ref{rhp:FN}, in which case the formula \eqref{eq:pm-BB} retains sense locally via a limit process (i.e., l'H\^{o}pital's rule).

\section{Asymptotic behavior of the rational Painlev\'e-II functions}\label{sec:asymptotics}
\subsection{Numerical observations and heuristic analysis}\label{sec:heuristics}
In this section, we assume without loss of generality that $m\ge 0$.
There have been several studies of the rational solutions $p_m(y)$ of the Painlev\'e-II equation from the numerical point of view, mostly concerned with looking for patterns in the distribution of poles of $p_m(y)$ in the complex $y$-plane as $m$ varies. The earliest work in this direction that we are aware of is the 1986 paper of Kametaka et al.~\cite{KametakaNFH86} in which numerical methods were brought to bear on the problem of f\/inding roots of the Yablonskii--Vorob'ev polynomials for $m$ as large as $m=37$; the f\/igures in~\cite{KametakaNFH86} for the largest values of $m$ display features suggesting the breakdown of the numerical method. A f\/igure such as those from~\cite{KametakaNFH86} also appears in the 1991 monograph~\cite{IwasakiKSY91}. These studies show the poles of $p_m(y)$ being contained for reasonably large $m$ within a roughly triangular-shaped region of size increasing with $m$ and therein organized in an apparently regular, crystalline pattern. Plots of poles of $p_m(y)$ obtained by similar methods also appear in~\cite{ClarksonM03}, a~paper that includes in addition a study of corresponding phenomena in higher-order equations in the Painlev\'e-II hierarchy. More recently, general numerical methods for the study of solutions with many poles in dif\/ferential equations have been advanced based on such techniques as Pad\'e approximation, and these methods have been shown to be capable of accurately reproducing the pole pattern of~$p_m(y)$, treating the Painlev\'e-II equation~\eqref{eq:myPII} as an initial-value problem to be solved numerically taking as initial conditions the exact values of~$p_m(0)$ and $p_m'(0)$ \cite{FornbergW14,Novokshenov14}. In Fig.~\ref{fig:poles} we give our own plots of poles of $p_m(y)$ for $m=15$, $m=30$, and $m=60$, which we made by symbolically constructing the relevant Yablonskii--Vorob'ev polynomials in \textit{Mathematica} and using \texttt{NSolve} with the option $\texttt{WorkingPrecision->50}$ to f\/ind the roots.
\begin{figure}[t]\centering
\includegraphics[width=0.3\linewidth]{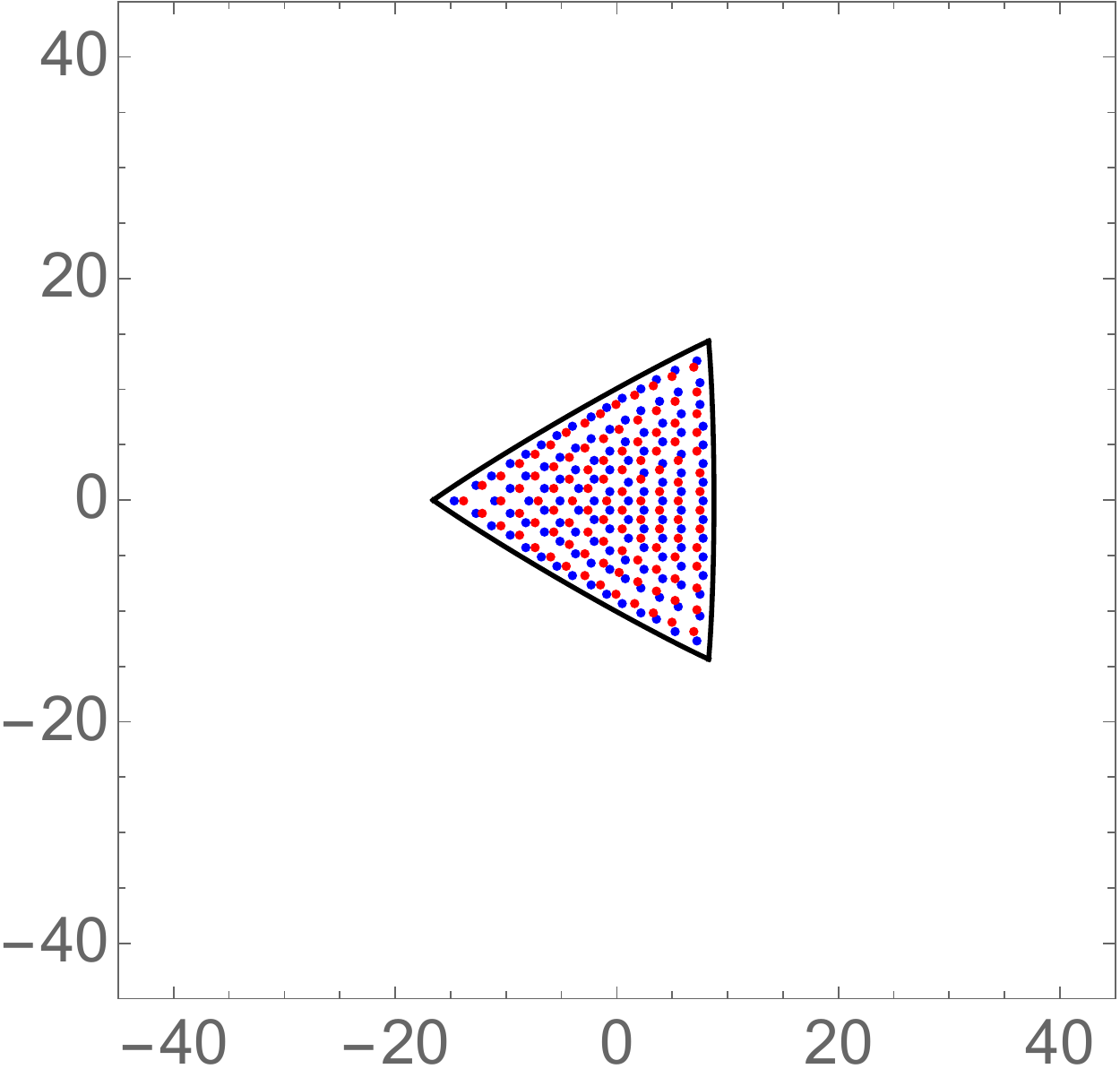}\hspace{0.03\linewidth}%
\includegraphics[width=0.3\linewidth]{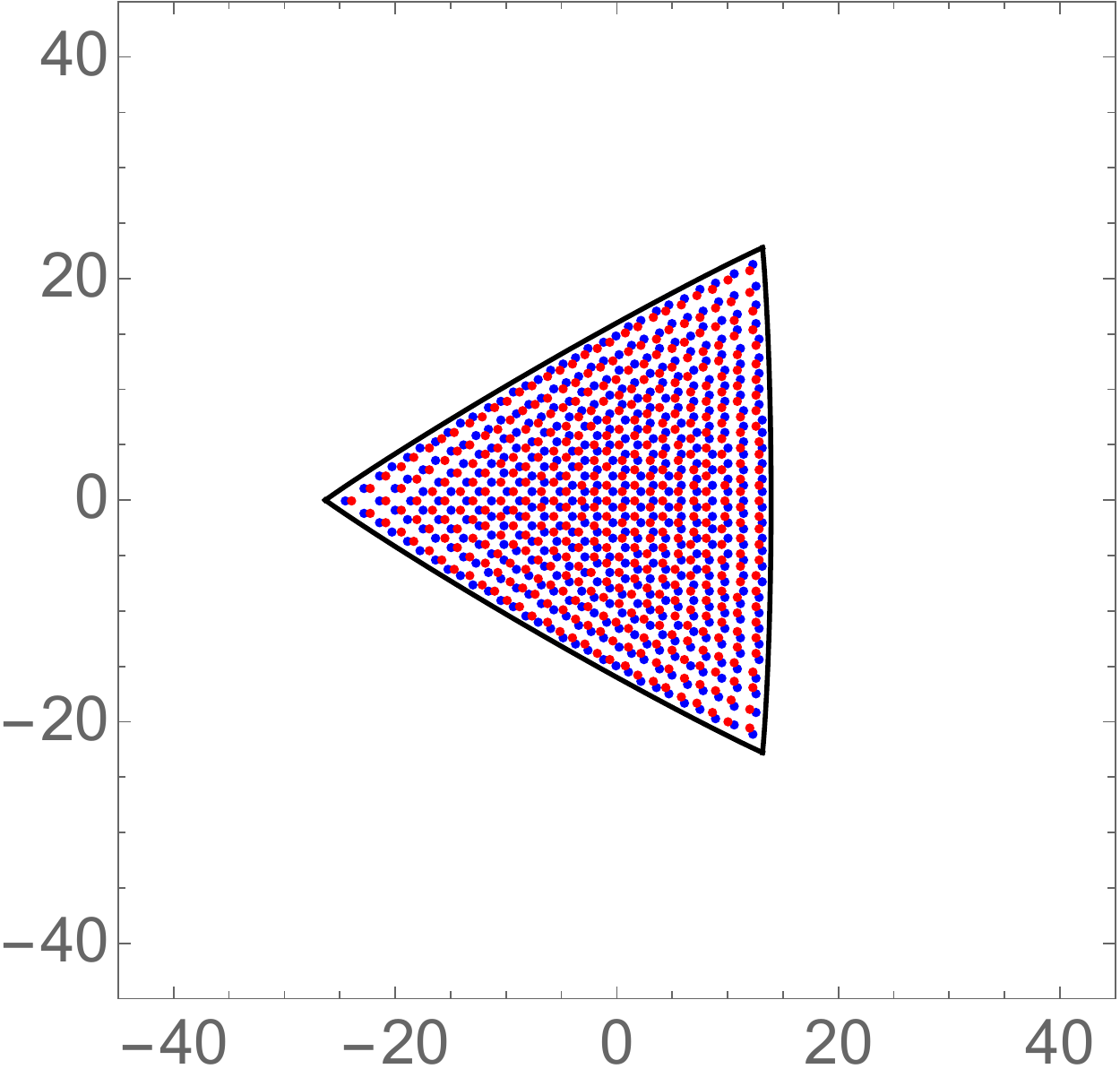}\hspace{0.03\linewidth}%
\includegraphics[width=0.3\linewidth]{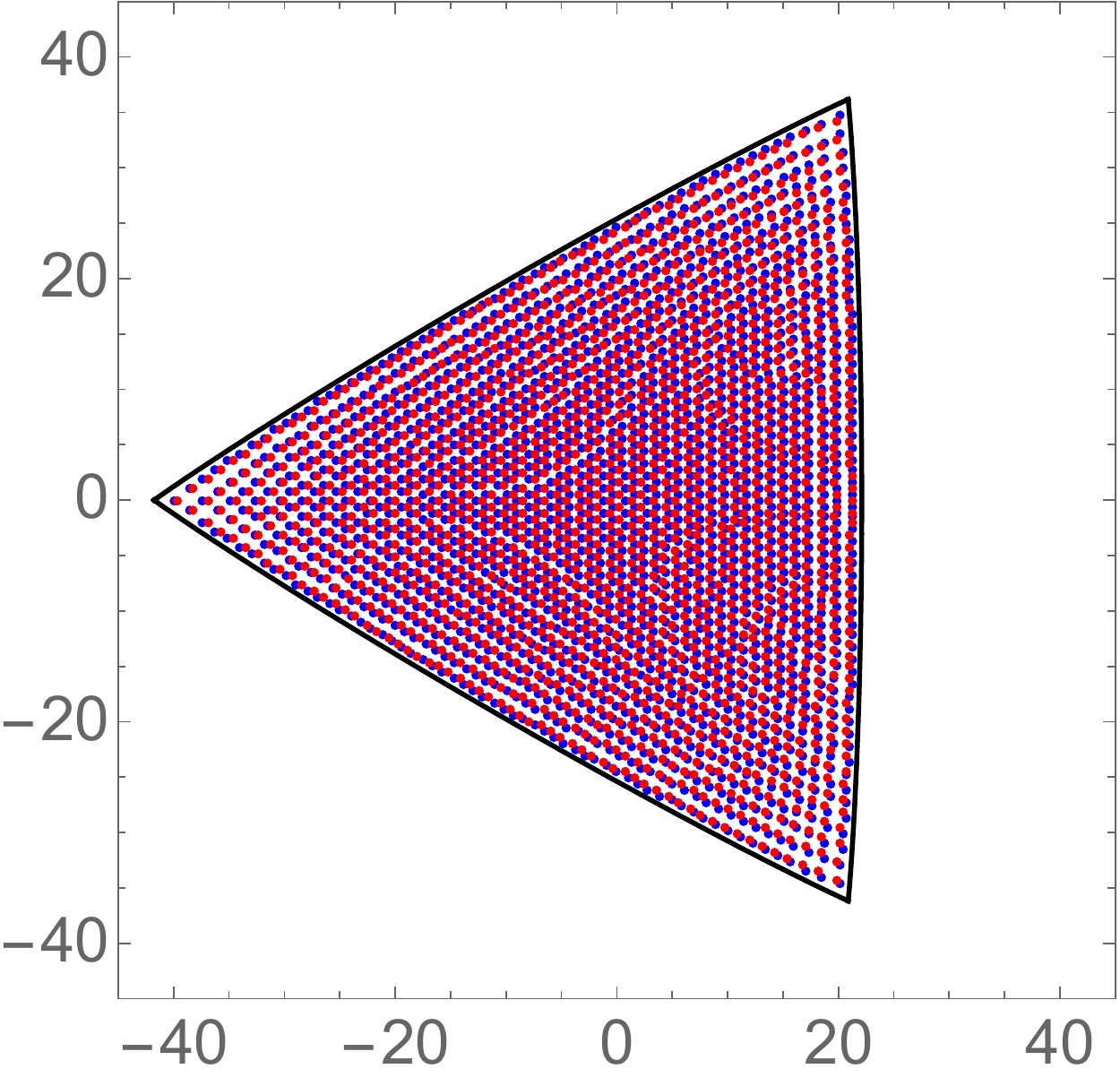}
\caption{The poles of residue $1$ (blue) and $-1$ (red) of $p_{15}(y)$ (left), $p_{30}(y)$ (center), and $p_{60}(y)$ (right). Superimposed is the theoretical boundary of the elliptic region (cf.\ Section~\ref{sec:elliptic-region}).}\label{fig:poles}
\end{figure}

These numerical observations suggest structure that should be explained, and yet the \mbox{large-$m$} limit in which the structural features of interest appear to become clear in the numerics is fundamentally out of reach of exact methods like iterated B\"acklund transformations or explicit determinantal formulae, the study of which becomes combinatorially prohibitive in this limit. Therefore one may consider instead methods of asymptotic analysis. A formal approach may be based upon the observation that the modulus of the poles or zeros of~$p_m(y)$ most distant from the origin scales roughly like $m^{2/3}$ \cite{IwasakiKSY91}, which suggests examining $p_m(y)$ in a small neighborhood of a~point $y=m^{2/3}x$; dominant balance arguments suggest that the size of the neighborhood should then be proportional to $m^{-1/3}$. So, letting $x\in\mathbb{C}$ be f\/ixed, consider the change of independent variable $y\mapsto w$ in~\eqref{eq:myPII} given by (the relatively small shifts by $1/2$ are convenient for later but at this point are inconsequential)
\begin{gather*}
y=\big(m-\tfrac{1}{2}\big)^{2/3}x+\big(m-\tfrac{1}{2}\big)^{-1/3}w.
\end{gather*}
Substituting this into \eqref{eq:myPII} along with the scaling of the independent variable by $p=\big(m-\tfrac{1}{2}\big)^{1/3}\mathcal{P}$, one arrives at the equivalent equation
\begin{gather*}
\frac{\dd^2\mathcal{P}}{\dd w^2}=2\mathcal{P}^3+\frac{2x}{3}\mathcal{P}-\frac{2}{3}+\frac{2w\mathcal{P}-1}{3\big(m+\tfrac{1}{2}\big)},
\end{gather*}
which for large $m$ appears to be a perturbation of an autonomous equation for an approximating function ${\widetilde{\mathcal{P}}}(w)$:
\begin{gather}
\frac{\dd^2{\widetilde{\mathcal{P}}}}{\dd w^2}=2{\widetilde{\mathcal{P}}}^3+\frac{2x}{3}{\widetilde{\mathcal{P}}}-\frac{2}{3}.
\label{eq:elliptic}
\end{gather}
Multiplying by $\dd{\widetilde{\mathcal{P}}}/\dd w$ and integrating gives
\begin{gather}
\left(\frac{\dd{\widetilde{\mathcal{P}}}}{\dd w}\right)^2={\widetilde{\mathcal{P}}}^4 + \frac{2x}{3}{\widetilde{\mathcal{P}}}^2-\frac{4}{3}{\widetilde{\mathcal{P}}}+\Pi,\label{eq:quartic}
\end{gather}
where $\Pi$ is an integration constant. If $\Pi$ and $x$ are related in such a way that the quartic polynomial on the right-hand side of \eqref{eq:quartic} has a double root ${\widetilde{\mathcal{P}}}_0$, then ${\widetilde{\mathcal{P}}}(w)={\widetilde{\mathcal{P}}}_0$ is an equilibrium solution of \eqref{eq:elliptic}. Double roots ${\widetilde{\mathcal{P}}}_0$ are necessarily related to $x$ via the cubic equation
\begin{gather}
3{\widetilde{\mathcal{P}}}_0^3+x{\widetilde{\mathcal{P}}}_0-1=0\label{eq:cubic}
\end{gather}
and then the relation between $\Pi$ and $x$ guaranteeing the existence of the double root can be expressed in terms of a solution ${\widetilde{\mathcal{P}}}_0={\widetilde{\mathcal{P}}}_0(x)$ of \eqref{eq:cubic} by
\begin{gather}
\Pi=\Pi_0(x):=2{\widetilde{\mathcal{P}}}_0(x)-\frac{2x}{3}{\widetilde{\mathcal{P}}}_0(x)^2.\label{eq:Pi0}
\end{gather}
It turns out (see Section~\ref{sec:outside} below) that this approximation of $\mathcal{P}(w)$ by the equilibrium solution ${\widetilde{\mathcal{P}}}_0(x)$ accurately describes the rational Painlev\'e-II function $p_m(y)$ in the pole-free region, provided that one selects the (unique) solution ${\widetilde{\mathcal{P}}}_0(x)$ of \eqref{eq:cubic} with the asymptotic behavior ${\widetilde{\mathcal{P}}}_0(x)=x^{-1}+\mathcal{O}\big(x^{-2}\big)$ as $x\to\infty$. This solution has branch points at $x=x_\mathrm{c}$ and $x=x_\mathrm{c}\ee^{\pm 2\pi\ii/3}$ for $x_\mathrm{c}:=-(9/2)^{2/3}$, which correspond to the corners of the triangular-shaped region containing the poles. More general solutions of~\eqref{eq:elliptic} can be expressed as elliptic functions of $w$ with elliptic modulus depending on the parameters $x$ and $\Pi$. These also turn out to be important in describing the rational Painlev\'e-II functions in the interior of the triangular region. Indeed, if one f\/ixes a~value of $x\in\mathbb{C}$ suf\/f\/iciently small to correspond to $y$ in the triangular region and views the rational Painlev\'e-II functions $p_m(y)$ as functions of the variable~$w$, one sees increasingly regular patterns of poles in the limit $m\to\infty$ suggestive of the period parallelogram of an elliptic function of~$w$. See Fig.~\ref{fig:zoomed-poles-and-zeros}.
\begin{figure}[t]\centering
\includegraphics[width=0.3\linewidth]{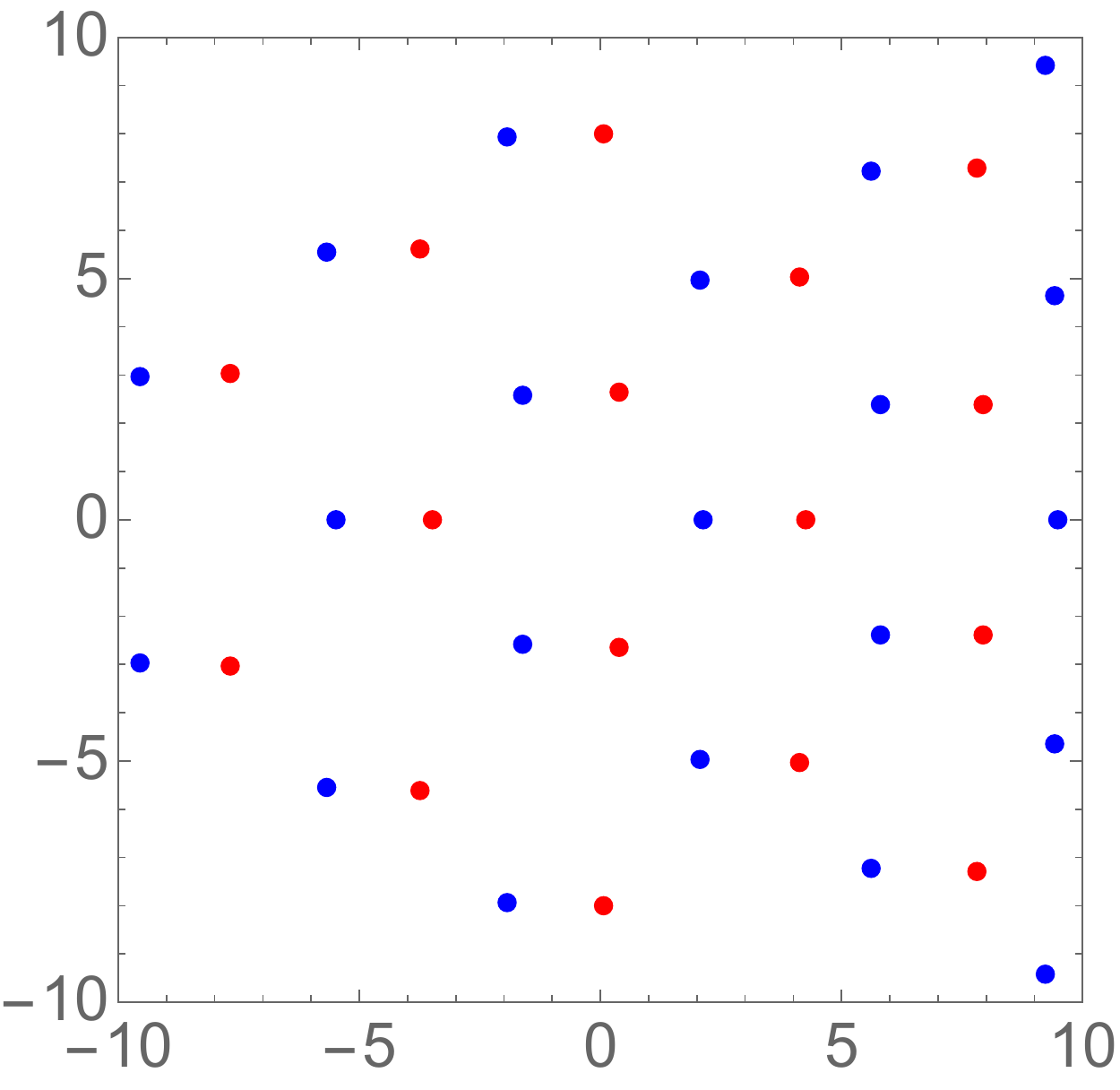}\hspace{0.03\linewidth}
\includegraphics[width=0.3\linewidth]{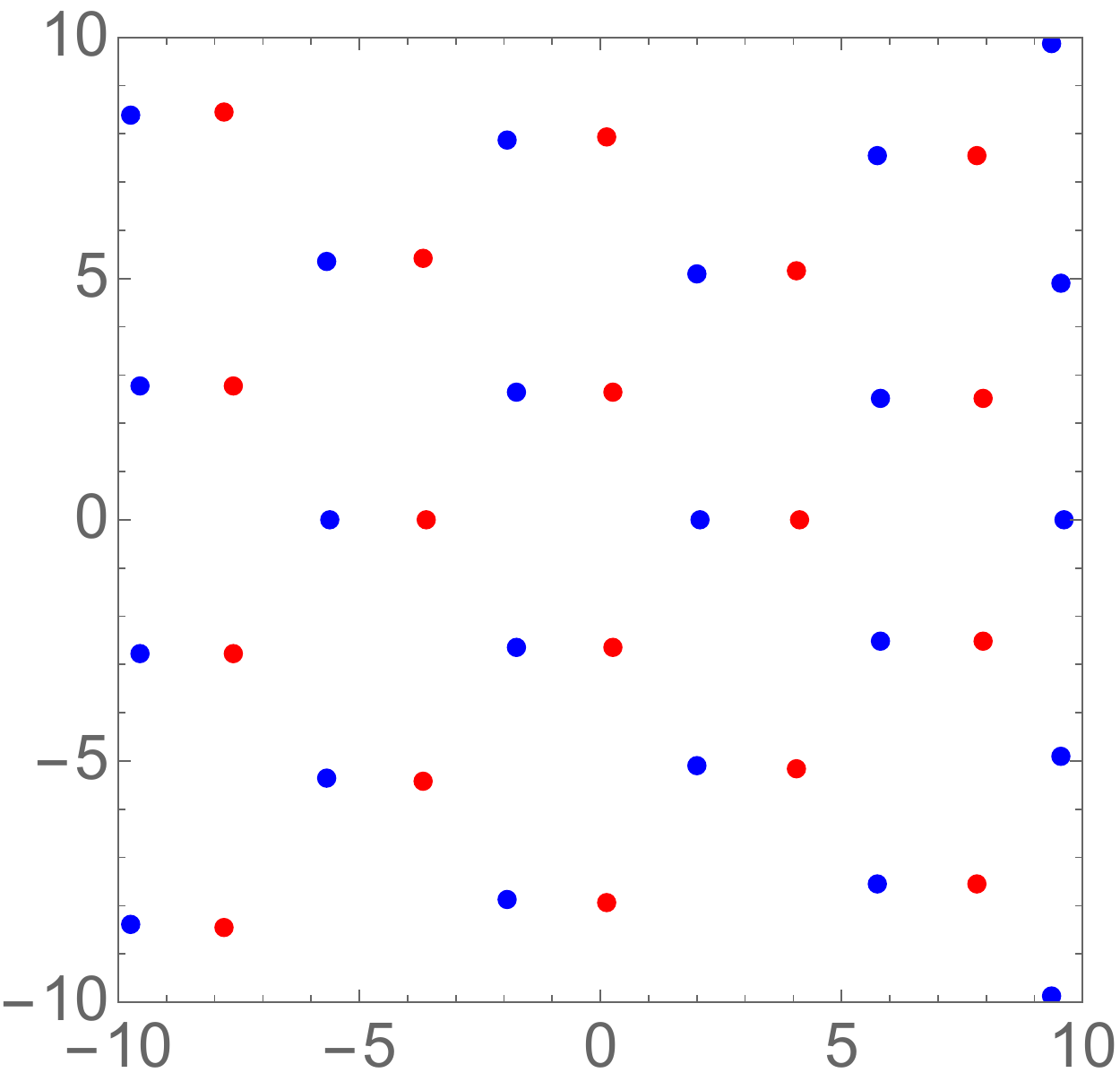}\hspace{0.03\linewidth}
\includegraphics[width=0.3\linewidth]{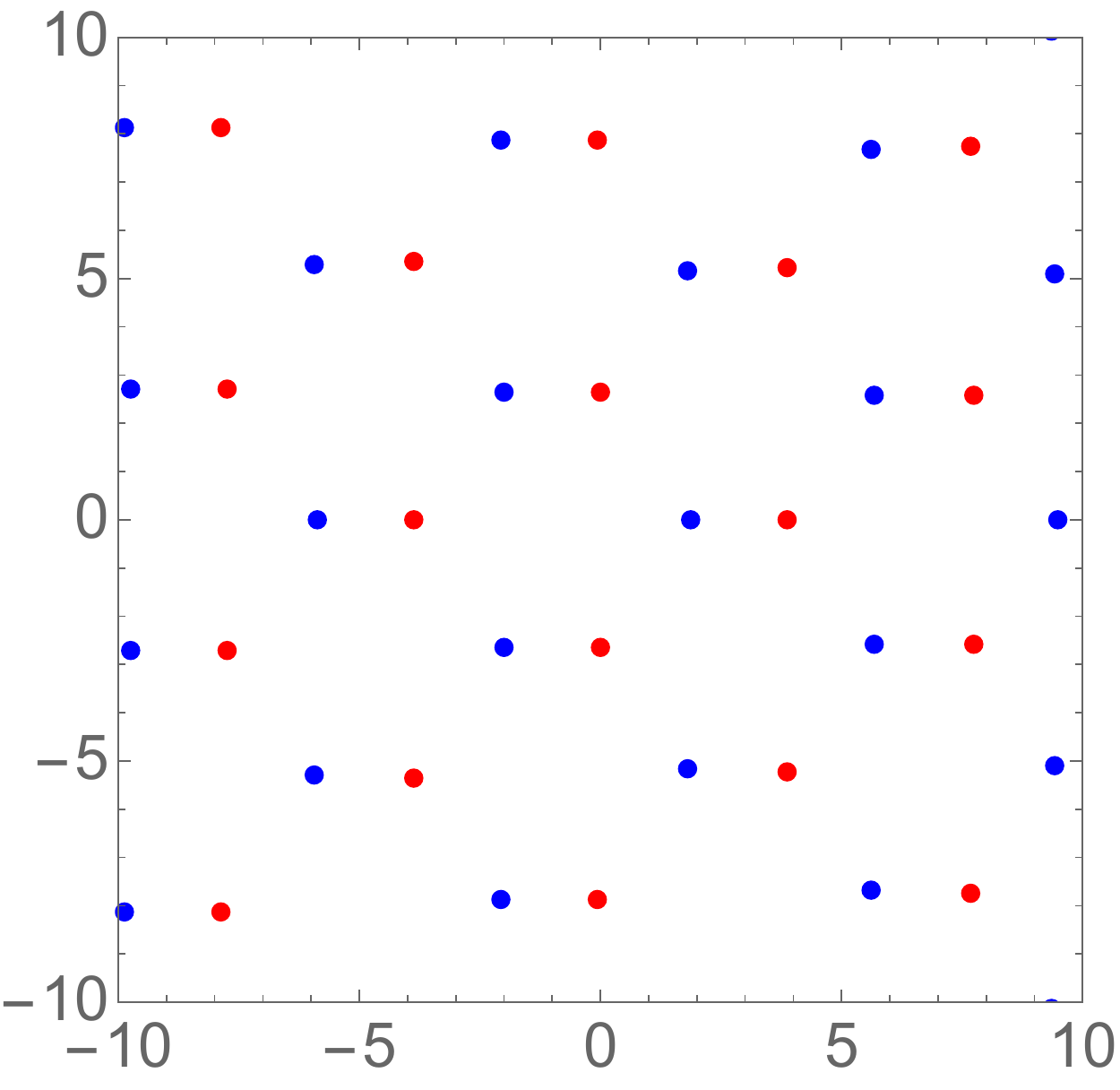}

\caption{The poles of residue $1$ (blue) and $-1$ (red) of $p_{m}(y)$ for $m=15$ (left), $m=30$ (center), and $m=60$ (right), plotted in the $w$-plane, a zoomed-in coordinate near $y=(m-\tfrac{1}{2})^{2/3}x$ for $x=-3/2$.}\label{fig:zoomed-poles-and-zeros}
\end{figure}
A similar formal scaling argument can be applied to study the asymptotic behavior of $p_m(y)$ near the corner points of the triangular region. For example, to zoom in on the corner point on the negative real axis, we may make the scalings
\begin{gather*}
p=-\left(\frac{m}{6}\right)^{1/3} - \left(\frac{128}{243 m}\right)^{1/15}Y\qquad\text{and}\qquad y=x_\mathrm{c}m^{2/3}+\left(\frac{243}{2m^2}\right)^{1/15}t,
\end{gather*}
after which one sees that the Painlev\'e-II equation \eqref{eq:myPII} takes the form
\begin{gather*}
\frac{\dd^2Y}{\dd t^2}=6Y^2+t + \mathcal{O}\big(m^{-2/5}\big)
\end{gather*}
for $t$ and $Y$ bounded, i.e., a perturbation of the Painlev\'e-I equation. This is a well-known degeneration of the Painlev\'e-II equation \cite{Kapaev97,KapaevK95}, and it suggests that particular solutions of the Painlev\'e-I equation may play a role in the asymptotic description of $p_m(y)$ near the three corner points. This also turns out to be true (see Section~\ref{sec:corner}).

\subsection{The elliptic region and its boundary}\label{sec:elliptic-region}
Let ${\widetilde{\mathcal{P}}}_0(x)$ denote the solution of the cubic equation \eqref{eq:cubic} with ${\widetilde{\mathcal{P}}}_0(x)= x^{-1}+\mathcal{O}\big(x^{-2}\big)$ as $x\to\infty$, which can be analytically continued to a maximal domain $\mathcal{D}$ consisting of the complex $x$-plane omitting three line segments connecting the three points $x_\mathrm{c}$, $\ee^{\pm 2\pi\ii/3}x_\mathrm{c}$ with the origin. For $x\in\mathcal{D}$, let $r(\kappa;x)$ denote the function def\/ined to satisfy $r(\kappa;x)^2= \kappa^2+2{\widetilde{\mathcal{P}}}_0(x) \kappa+{\widetilde{\mathcal{P}}}_0(x)^2-\tfrac{2}{3}{\widetilde{\mathcal{P}}}_0(x)^{-1}$ and $r(\kappa;x)={\kappa}+\mathcal{O}(1)$ as $\kappa\to\infty$, def\/ined on a maximal domain of analyticity in the $\kappa$-plane\footnote{The complex variable $\kappa$ (written as $z$ in \cite{BuckinghamM14,BuckinghamM15}) is a rescaling of the variable $\zeta$ from Riemann--Hilbert Problem~\ref{rhp:JM}.} omitting only the segment connecting the roots of $r(\kappa;x)^2$, one of which we denote by~$a(x)$. We def\/ine a function $\mathfrak{c}(x)$ by
\begin{gather}
\mathfrak{c}(x):=\frac{3}{2}\int_{a(x)}^{{\widetilde{\mathcal{P}}}_0(x)}(\kappa-{\widetilde{\mathcal{P}}}_0(x))r(\kappa;x)\,\dd \kappa,\qquad x\in\mathcal{D}, \label{eq:c-define}
\end{gather}
where the path of integration is arbitrary\footnote{It can be checked that the value of $\mathfrak{c}(x)$ is unchanged by adding loops around the branch cut of $r(\kappa;x)$ to the path of integration because ${\widetilde{\mathcal{P}}}_0(x)$ satisf\/ies~\eqref{eq:cubic}.} within the domain of analyticity of~$r(\kappa;x)$.

It turns out that in the limit $m\to\infty$, the region of the complex plane that contains the poles of $p_m(y)$ is $y\in m^{2/3}T$, where $T$ is the bounded component of the set of $x\in\mathbb{C}$ for which $\operatorname{Re}(\mathfrak{c}(x))\neq 0$. The boundary $\partial T$ consists of points for which $\operatorname{Re}(\mathfrak{c}(x))=0$. The integral in~\eqref{eq:c-define} can be evaluated in terms of elementary functions, taking appropriate care of branches of multivalued functions; expressions can be found in~\cite{BertolaB15,BuckinghamM14}. The exact formula is less important than the basic property that $\mathfrak{c}(x)$ is analytic for $x\in\mathcal{D}$ with algebraic branch points at the points $x=x_\mathrm{c}$ and $x=x_\mathrm{c}\ee^{\pm 2\pi\ii/3}$. This implies that $\partial T$ is a union of three analytic arcs joining the branch points pairwise, with ref\/lection symmetry in the real axis and rotation symmetry about the origin by integer multiples of $2\pi/3$. The curve $m^{2/3}\partial T$ is superimposed on each of the pole plots in Fig.~\ref{fig:poles}. We call $T$ the \emph{elliptic region}, the three branch points of ${\widetilde{\mathcal{P}}}_0(x)$ its \emph{corners}, and the three smooth arcs of $\partial T$ its \emph{edges}. Local analysis of $\mathfrak{c}(x)$ shows \cite[Section~2.3]{BuckinghamM15} that the interior angles of $\partial T$ at the three corners are all $2\pi/5$, so that $\partial T$ is a ``curvilinear triangle'' at best.

\subsection[Asymptotic description of $p_m(y)$ by steepest descent]{Asymptotic description of $\boldsymbol{p_m(y)}$ by steepest descent}
We now present several results on the asymptotic behavior of the rational Painlev\'e-II function~$p_m(y)$, all of which have been obtained by the application of variants of the Deift--Zhou steepest descent method \cite{DeiftZ93} to either Riemann--Hilbert Problem~\ref{rhp:JM} (see \cite{BuckinghamM14,BuckinghamM15}) or Riemann--Hilbert Problem~\ref{rhp:BB} (see~\cite{BertolaB15}). Regardless of which Riemann--Hilbert problem is the starting point, the basic steps of the method are the same:
\begin{enumerate}\itemsep=0pt
\item Introduce a diagonal matrix multiplier built from exponentials of a scalar function frequently called a ``$g$-function'' with the aim of simultaneously obtaining normalization to the identity matrix at inf\/inity and stabilizing the jump matrices of the problem so that they are alternately exponentially small perturbations of either constant matrices or purely oscillatory matrices along dif\/ferent contour arcs. Frequently this step also requires some deformation of the contour of the original Riemann--Hilbert problem by means of analytic continuation of the jump matrices.
\item Use explicit matrix factorizations to algebraically separate oscillatory factors in the jump matrices having phase derivatives of opposite signs. Splitting the jump contour into separate arcs for each factor, a subsequent deformation to either side of the original jump contour ensures that the oscillatory factors now become exponentially small in the limit $m\to\infty$.
\item Construct an explicit model of the solution called a ``parametrix'' by considering only those remaining jump matrices that are not exponentially small perturbations of the identity matrix.
\item By comparing the unknown matrix obtained after the second step with the parametrix, obtain an equivalent Riemann--Hilbert problem for the matrix quotient. The aim of the method is to ensure that the resulting Riemann--Hilbert problem is of ``small-norm'' type, meaning that it can be solved by a convergent iterative procedure that also allows for the rigorous estimation of the solution. This analysis proves the accuracy of approximate formulae for the unknowns of interest, such as $p_m(y)$, which are extracted from the explicit parametrix.
\end{enumerate}
The steepest descent method gets its name from the second step in the procedure, which resembles the type of contour deformations that one carries out in implementing the steepest descent method for the asymptotic expansion of exponential integrals.

The form of the parametrix that one obtains is determined in most of the complex plane by the number of contour arcs on which the $g$-function induces oscillations. This number is related to the genus of a hyperelliptic Riemann surface whose function theory is exploited to construct the parametrix. As the original Riemann--Hilbert problem depends on a complex parameter $y$, it is to be expected that the genus may be dif\/ferent for dif\/ferent values of $y\in\mathbb{C}$, leading to the phenomenon of phase transitions. Indeed, the boundary of the elliptic region turns out to be exactly such a phase transition. In particular the hyperelliptic curve that characterizes the rational Painlev\'e-II function $p_m(y)$ for large~$m$ when $y$ lies outside of the elliptic region has genus zero. An interesting dif\/ference between the application of the steepest descent method to the Jimbo--Miwa problem \cite{BuckinghamM14,BuckinghamM15} and its application to the Bertola--Bothner problem~\cite{BertolaB15} is that in the former case the curve corresponding to the elliptic region has genus $1$ (hence the terminology ``elliptic'') while in the latter case it instead has genus $2$ (with some symmetries that allow its function theory to be reducible to elliptic functions after all, see \cite[Section~4.6]{BertolaB15}).

We give no further details of the proofs of the following results, leading the reader to the original references \cite{BertolaB15,BuckinghamM14,BuckinghamM15} for complete information. We also note that some of the results below have also been captured by the isomonodromy method, a WKB-ansatz based asymptotic approach to Riemann--Hilbert problems \cite{Kapaev97}.

\subsubsection[Asymptotic description of $p_m$ in the exterior region]{Asymptotic description of $\boldsymbol{p_m}$ in the exterior region}\label{sec:outside}
The simplest result to state is the following.
\begin{Theorem}[Buckingham \& Miller \protect{\cite[Theorem~1]{BuckinghamM14}}, Bertola \& Bothner \protect{\cite[Corollary~6.1]{BertolaB15}}]\label{thm:outside} Given a sufficiently large integer $m>0$, let $K_m$ be a set of points $x$ in the exterior of $T$ uniformly bounded away from the corners but otherwise with $\operatorname{dist}(x,T)>\ln(m)/m$. Then the rational Painlev\'e-II function $p_m(y)$ satisfies
\begin{gather*}
m^{-1/3}p_m\big(m^{2/3}x\big)={\widetilde{\mathcal{P}}}_0(x) + \mathcal{O}\big(m^{-1}\big),\qquad m\to\infty
\end{gather*}
with the error term being uniform for $x\in K_m$. In particular, $p_m\big(m^{2/3}x\big)$ is pole free for $x\in K_m$ and $m$ sufficiently large.
\end{Theorem}
Recall that the limiting function ${\widetilde{\mathcal{P}}}_0(x)$ also has an interpretation as an equilibrium (``fast'' variable $w$-independent) solution of the formal model dif\/ferential equation~\eqref{eq:elliptic}. In~\cite{BuckinghamM14} this result is reported with an unimportant shift of the scaling parameter $m\mapsto m-\tfrac{1}{2}$ in the argument of~$p_m$, as this was convenient for the Riemann--Hilbert analysis used to prove the theorem. Once~$x$ moves into the elliptic region $T$ and wild oscillations develop, this shift will have to be retained to ensure full accuracy.

\subsubsection[Asymptotic description of $p_m$ in the elliptic region]{Asymptotic description of $\boldsymbol{p_m}$ in the elliptic region}
Now considering $x\in T$, we def\/ine the integration constant $\Pi$ in \eqref{eq:quartic} no longer via \eqref{eq:Pi0} but rather via the following \emph{Boutroux conditions}:
\begin{gather}
\operatorname{Re}\left(\oint\nolimits_\mathfrak{a}\frac{\dd{\widetilde{\mathcal{P}}}}{\dd w}\,\dd{\widetilde{\mathcal{P}}}\right)=0\qquad\text{and}\qquad
\operatorname{Re}\left(\oint\nolimits_\mathfrak{b}\frac{\dd{\widetilde{\mathcal{P}}}}{\dd w}\,\dd{\widetilde{\mathcal{P}}}\right)=0,
\label{eq:Boutroux}
\end{gather}
where $(\mathfrak{a},\mathfrak{b})$ is a basis of homology cycles on the elliptic curve $\Gamma(x)$ determined as a subvariety of $\mathbb{C}^2$ with coordinates $({\widetilde{\mathcal{P}}},\dd{\widetilde{\mathcal{P}}}/\dd w)$ given by \eqref{eq:quartic}. In \cite[Proposition~5]{BuckinghamM14} it is shown that these conditions determine $\Pi=\Pi(x)$ uniquely as a continuous function on $T$ with $\Pi(0)=0$. Moreover, the four roots of the polynomial on the right-hand side of \eqref{eq:quartic} are then distinct for $x\in T$, with two roots degenerating when $x$ approaches an edge point of $\partial T$ and all four roots degenerating when $x$ approaches a corner point of $\partial T$. The function $\Pi(x)$ determined from the Boutroux conditions \eqref{eq:Boutroux} is smooth but decidedly non-analytic in~$x$ (cf.\ \cite[equation~(4.31)]{BuckinghamM14}).

Given a point $x\in T$, we let $A(x)$, $B(x)$, $C(x)$, and $D(x)$ denote the roots of the quartic $R(\kappa;x)^2 =\kappa^4+\tfrac{2}{3}x\kappa^2 -\tfrac{4}{3}\kappa+\Pi(x)$, observing that the notation is well-def\/ined by continuity in~$x$ given that when $x=0$ the roots are as shown in Fig.~\ref{fig:elliptic-RHP}.
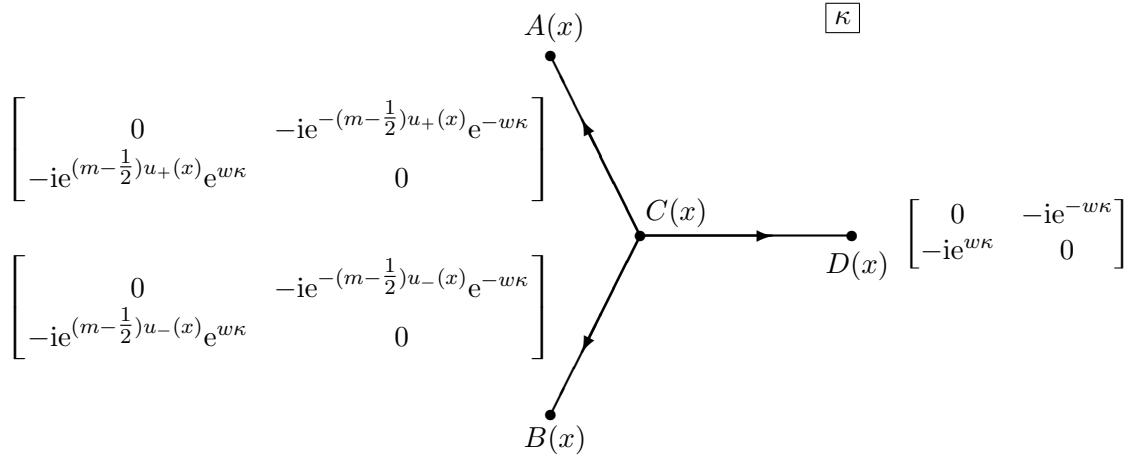
\begin{figure}[t]\centering
\setlength{\unitlength}{2pt}
\begin{picture}(100,100)(-50,-50)
\thicklines
\put(45,40){\framebox{$\kappa$}}
\put(10,0){\line(1,0){40}}
\put(10,0){\vector(1,0){25}}
\put(10,0){\line(-1,2){17}}
\put(10,0){\vector(-1,2){11}}
\put(10,0){\line(-1,-2){17}}
\put(10,0){\vector(-1,-2){11}}
\put(-7,34){\circle*{2}}
\put(-12,38){$A(x)$}
\put(-110,15){$\begin{bmatrix}0 & -\ii\ee^{-(m-\tfrac{1}{2})u_+(x)}\ee^{-w\kappa}\\-\ii\ee^{(m-\tfrac{1}{2})u_+(x)}\ee^{w\kappa} & 0\end{bmatrix}$}
\put(-7,-34){\circle*{2}}
\put(-12,-40){$B(x)$}
\put(-110,-15){$\begin{bmatrix}0 & -\ii\ee^{-(m-\tfrac{1}{2})u_-(x)}\ee^{-w\kappa}\\-\ii\ee^{(m-\tfrac{1}{2})u_-(x)}\ee^{w\kappa} & 0\end{bmatrix}$}
\put(10,0){\circle*{2}}
\put(11,3){$C(x)$}
\put(50,0){\circle*{2}}
\put(45,-7){$D(x)$}
\put(60,-1){$\begin{bmatrix}0 & -\ii \ee^{-w\kappa}\\-\ii\ee^{w\kappa} & 0\end{bmatrix}$}
\end{picture}

\caption{The branch cuts of $R(\kappa;x)$ for $x=0$ and the jump matrix $\mathbf{W}(\kappa;x,w)$ for Riemann--Hilbert Problem~\ref{rhp:model}.}\label{fig:elliptic-RHP}
\end{figure}
We then def\/ine $R(\kappa;x)$ as an analytic function satisfying $R(\kappa;x)=\kappa^2+\mathcal{O}(\kappa)$ as $\kappa\to\infty$ and with branch cuts along line segments connecting the four branch points as illustrated in Fig.~\ref{fig:elliptic-RHP}. Now def\/ine
\begin{gather*}
u_+(x):= 3\int_{D(x)}^{A(x)}R(\kappa;x)\,\dd \kappa\qquad\text{and}\qquad
u_-(x):=3\int_{D(x)}^{B(x)}R(\kappa;x)\,\dd \kappa,
\end{gather*}
where the path of integration is in each case assumed to be a straight line. In order to present the results for $x\in T$, we f\/irst formulate an auxiliary Riemann--Hilbert problem:
\begin{rhp}\label{rhp:model}
Let $x\in T$ and $w\in\mathbb{C}$ be given and let $m\ge 0$ be an integer. Seek a $2\times 2$ matrix-valued function $\mathbf{X}^m(\kappa;x,w)$ defined for $\kappa$ in the same domain where $R(\kappa;x)$ is analytic, with the following properties:
\begin{itemize}
\item\textit{\textbf{Analyticity.}} $\mathbf{X}^m(\kappa;x,w)$ is analytic in $\kappa$ in its domain of definition, taking continuous boundary values $\mathbf{X}^m_+(\kappa;x,w)$ and $\mathbf{X}^m_-(\kappa;x,w)$ from the left and right respectively on each oriented arc of its jump contour as shown in Fig.~{\rm \ref{fig:elliptic-RHP}}, \emph{except} at the four branch points where~$-1/4$ power singularities are admitted.
\item\textit{\textbf{Jump condition.}} The boundary values are related by
\begin{gather*}
\mathbf{X}_-^m(\kappa;x,w)=\mathbf{X}_+^m(\kappa;x,w)\mathbf{W}(\kappa;x,w),
\end{gather*}
where the jump matrix $\mathbf{W}(\kappa;x,w)$ is defined on each arc of the jump contour as shown in Fig.~{\rm \ref{fig:elliptic-RHP}}.
\item\textit{\textbf{Normalization.}} The matrix $\mathbf{X}^m(\kappa;x,w)$ is normalized at $\kappa=\infty$ as follows:
\begin{gather*}
\lim_{\kappa\to\infty}\mathbf{X}^m(\kappa;x,w)=\mathbb{I},
\end{gather*}
where the limit may be taken in any direction.
\end{itemize}
\end{rhp}
The matrix $\mathbf{X}^m(\cdot;x,w)$ is denoted $\dot{\mathbf{O}}^\mathrm{(out)}(\cdot)$ in~\cite{BuckinghamM14}. From the Laurent coef\/f\/icients
\begin{gather*}
\mathbf{X}^m_1(x,w):=\lim_{\kappa\to\infty} \kappa\big(\mathbf{X}^m(\kappa;x,w)-\mathbb{I}\big),\\
\mathbf{X}^m_2(x,w):=\lim_{\kappa\to\infty} \kappa^2\big(\mathbf{X}^m(\kappa;x,w)-\mathbb{I}-\mathbf{X}^m_1(x,w)\kappa^{-1}\big)
\end{gather*}
we then def\/ine a function ${\widetilde{\mathcal{P}}}^m(x,w)$ by
\begin{gather*}
{\widetilde{\mathcal{P}}}^m(x,w):=X_{1,22}^m(x,w)-\frac{X_{2,12}^m(x,w)}{X_{1,12}^m(x,w)}.
\end{gather*}
Then we have the following result.
\begin{Theorem}[Buckingham \& Miller \protect{\cite[Proposition 7 \& Theorem 2]{BuckinghamM14}}]\label{thm:inside}
For each $x\in T$ and integer $m\ge 0$, ${\widetilde{\mathcal{P}}}^m(x,w)$ is an elliptic function of $w$ that satisfies the model equation \eqref{eq:elliptic} $($more precisely, with $\Pi=\Pi(x)$ defined as above, equation \eqref{eq:quartic}$)$. Defining
\begin{gather*}
\chi^m(x,w):=\begin{cases}\hphantom{-}1,& |{\widetilde{\mathcal{P}}}^m(x,w)|\le 1,\\
-1,& |{\widetilde{\mathcal{P}}}^m(x,w)|>1,
\end{cases}
\end{gather*}
the asymptotic condition
\begin{gather}
m^{-\chi^m(x,w)/3}p_m(y)^{\chi^m(x,w)}={\widetilde{\mathcal{P}}}^m(x,w)^{\chi^m(x,w)} +\mathcal{O}\big(m^{-1}\big),\nonumber\\
 y=\big(m-\tfrac{1}{2}\big)^{2/3}x + \big(m-\tfrac{1}{2}\big)^{-1/3}w, \label{eq:inside-approximation}
\end{gather}
holds as $m\to\infty$ uniformly for $(x,w)$ in compact subsets of $T\times\mathbb{C}$.
\end{Theorem}

The statement \eqref{eq:inside-approximation} says\footnote{This statement corrects a mistake in equation (4.219) of~\cite{BuckinghamM14}. Equations~(4.217), (4.218), and (4.220) of that reference should be similarly reformulated.} that $m^{-1/3}p_m(y)$ and ${\widetilde{\mathcal{P}}}^m(x,w)$ are uniformly close where ${\widetilde{\mathcal{P}}}^m(x,w)$ is bounded, while their reciprocals are uniformly close where ${\widetilde{\mathcal{P}}}^m(x,w)$ is bounded away from zero. The fact that the approximating function ${\widetilde{\mathcal{P}}}^m(x,w)$ depends on two variables deserves some explanation. Since $w$ should be bounded for the indicated error estimate to be valid, variation of $w$ amounts to the exploration of a small neighborhood of radius $m^{-1/3}$ of the point $y=\big(m-\tfrac{1}{2}\big)^{2/3}x$. Thus f\/ixing $x\in T$ and varying $w$ one obtains a local approximation whose validity fails if $w$ becomes large. It is on the $w$-scale that $m^{-1/3}p_m(y)$ is well-approximated by an elliptic function of $w$, the meromorphic nature of which mirrors that of the original rational Painlev\'e-II function $p_m(y)$. On the other hand, the same approximating formula \eqref{eq:inside-approximation} also allows~$x$ to vary within~$T$; here one may f\/ix arbitrarily, say, $w=0$ and obtain an approximation that is uniformly valid on compact subsets of $T$ that avoid poles, but that has an essentially non-meromorphic character due to the nonanalyticity of~$\Pi(x)$. Geometrically, we may view~$T$ as a~manifold with base coordinate~$x$, while~$w$ plays the role of a coordinate on the tangent space to~$T$ at~$x$. Thus \eqref{eq:inside-approximation} approximates $p_m(y)$ with a function ${\widetilde{\mathcal{P}}}^m(x,w)$ def\/ined on the tangent bundle to~$T$. We also can call $x$ a \emph{macroscopic} variable and $w$ a \emph{microscopic} variable to distinguish their dif\/ferent roles in~\eqref{eq:inside-approximation}.

Numerous auxiliary results can be obtained from Theorem~\ref{thm:inside}. Perhaps the main quantity of interest is the distribution of poles of residues $\pm 1$, which by \eqref{eq:inside-approximation} form regular lattices of spacing proportional to~$m^{-1/3}$ in the $y$-variable that slowly vary over distances proportional to~$m^{2/3}$ (the macroscopic $x$-scale) in the same variable. Bertola and Bothner characterize each lattice globally via a pair of quantization conditions giving the lattice points as the intersections of two distinct families curves over~$T$. In \cite[Proposition~14]{BuckinghamM14} it is shown that, while the period parallelograms of the lattices have limits in the $w$-plane as~$m\to\infty$ for given $x\in T$, the of\/fset of the lattices in the $w$-plane can f\/luctuate with $m$, accumulating a f\/ixed shift with each increment of~$m$ by a vector depending on the base point $x\in T$; see Fig.~\ref{fig:zoom-origin}.
\begin{figure}[t]\centering
\includegraphics[width=0.3\linewidth]{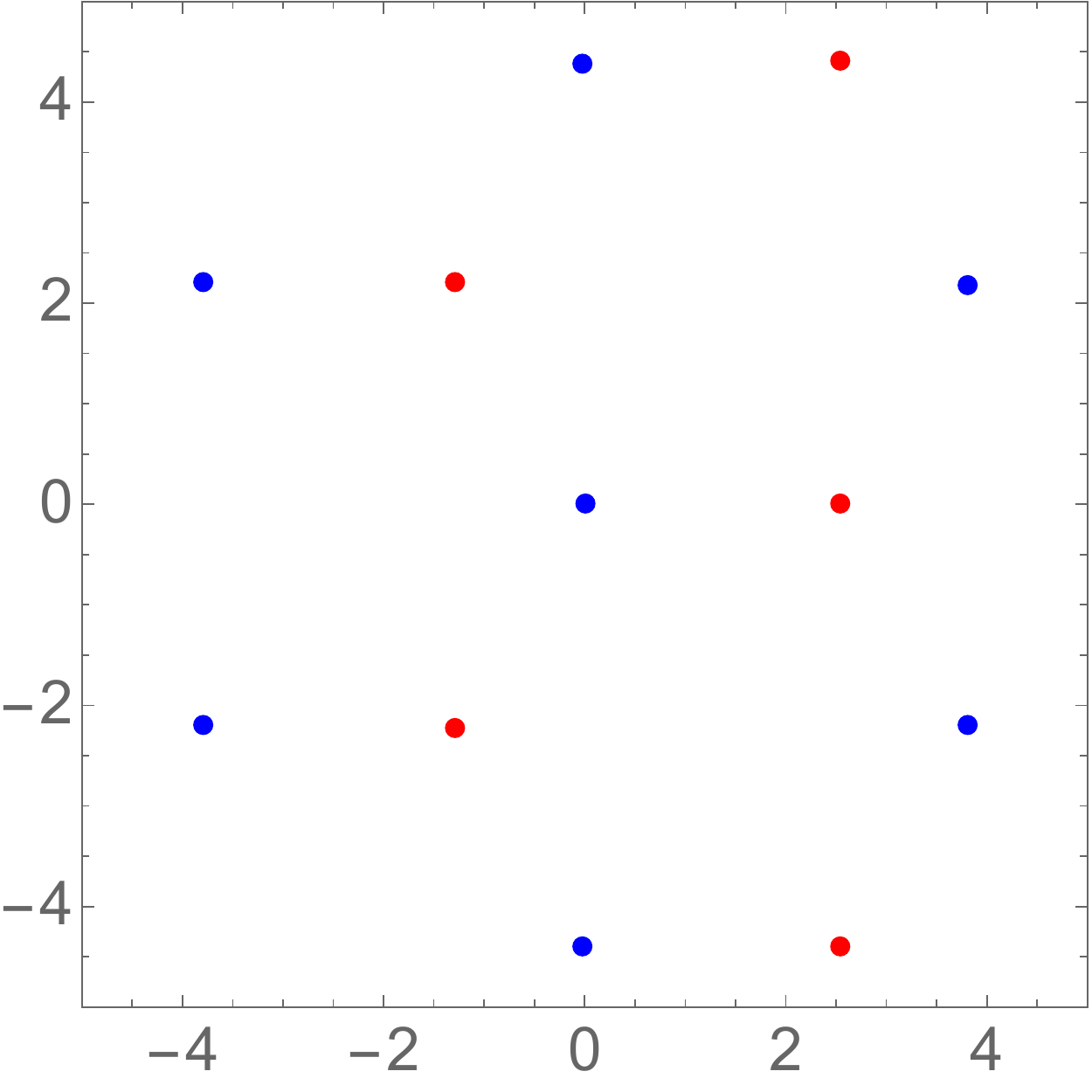}\hspace{0.03\linewidth}%
\includegraphics[width=0.3\linewidth]{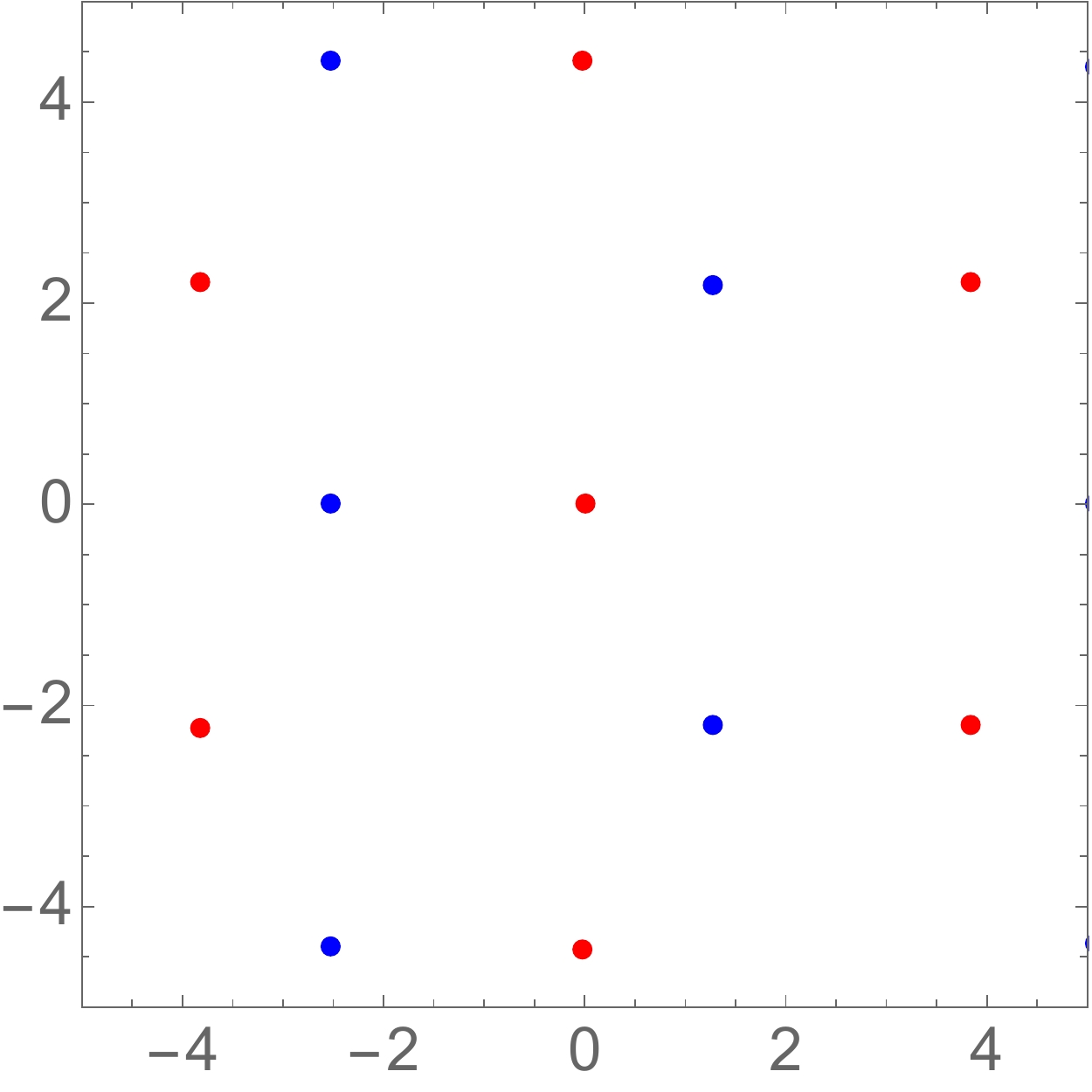}\hspace{0.03\linewidth}%
\includegraphics[width=0.3\linewidth]{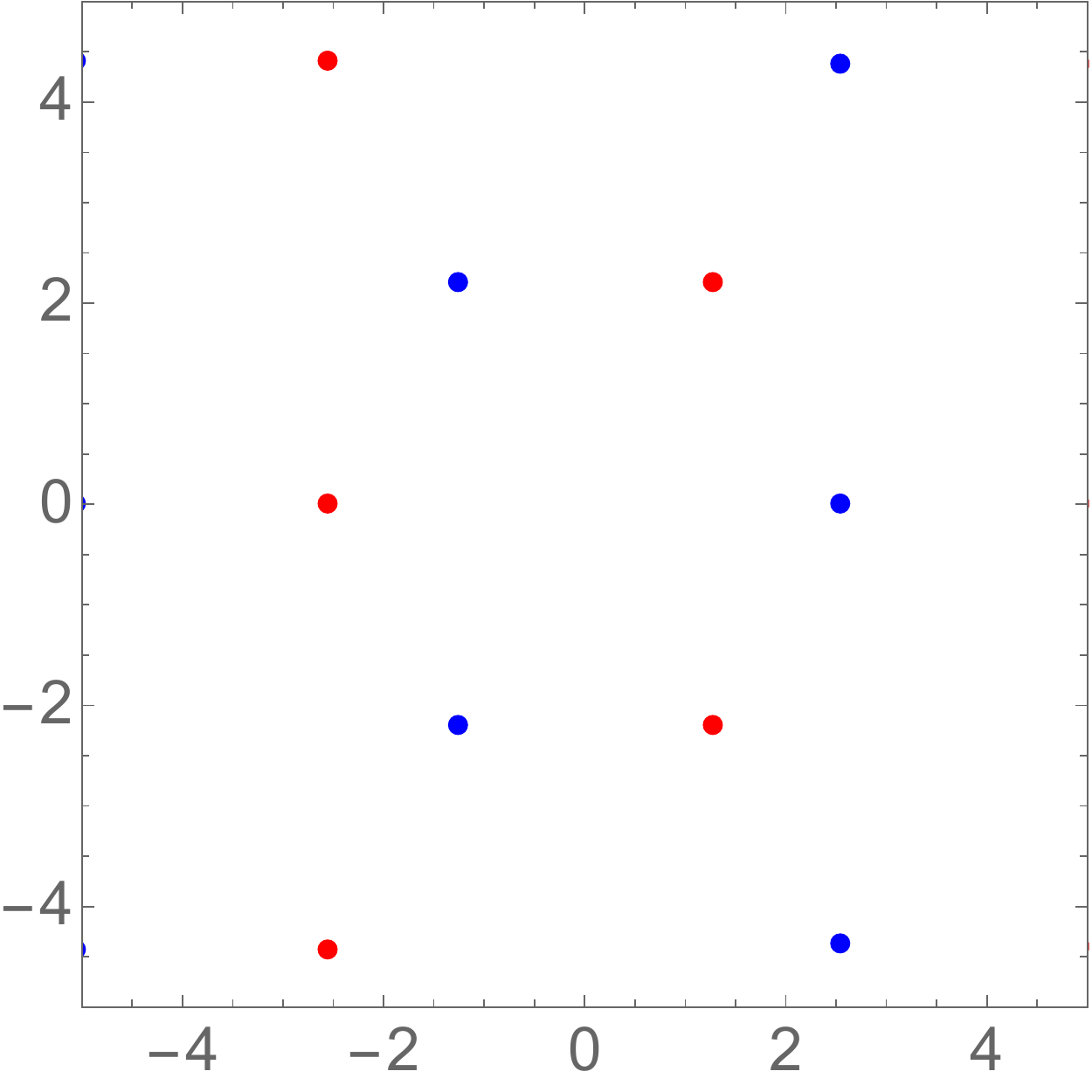}
\caption{The poles of residue $1$ (blue) and $-1$ (red) of $p_{m}(y)$ for $m=58$ (left), $m=59$ (center), and $m=60$ (right), plotted in the $w$-plane for $x=0$. Note the shift of the lattices with $m$; when $x=0$, three consecutive shifts make up a lattice vector, so the asymptotic pattern has period $3$ with respect to $m$. This dependence of the microscopic pattern near $x=0$ on $m\pmod{3}$ has also been noted in a related problem by Shapiro and Tater \cite{ShapiroT14}.}\label{fig:zoom-origin}
\end{figure}
As for how accurately the lattice points approximate the poles of $p_m$, it can be proved that the true poles of $p_m\big(\big(m-\tfrac{1}{2}\big)^{2/3}x\big)$ lying in any compact subset of $T$ all move within the union of disks of radius of radius $\mathcal{O}\big(1/m^2\big)$ centered at the lattice points (whose spacing in $x$ is proportional to $1/m$) if $m$ is suf\/f\/iciently large \cite[Corollary~1]{BuckinghamM14}. See also \cite[Theorem~1.6]{BertolaB15}, where this result is formulated for disks of radius $o(1/m)$.

In~\cite{BuckinghamM14}, formulae are also given for the asymptotic density of poles of $p_m\big(\big(m-\tfrac{1}{2}\big)^{2/3}x\big)$ as a~function of $x\in T$. Here, density is measured in terms of the microscopic coordinate $w$, and one may def\/ine both a planar density:
\begin{gather*}
{\widetilde{\sigma}}_\mathrm{P}(x):=\lim_{M\uparrow\infty}\frac{\#\{\text{residue $-1$ poles $w$ of ${\widetilde{\mathcal{P}}}^m(x,w)$ with $|w|<M$}\}}{\pi M^2},\qquad x\in T,
\end{gather*}
and a linear density of real poles for $x\in T\cap\mathbb{R}$:
\begin{gather*}
{\widetilde{\sigma}}_\mathrm{L}(x):=\lim_{M\uparrow\infty}\frac{\#\{\text{real residue $-1$ poles $w$ of ${\widetilde{\mathcal{P}}}^m(x,w)$ in $(-M,M)$}\}}{2M},\qquad x\in T\cap\mathbb{R}.
\end{gather*}
Since there are precisely two simple poles of opposite residue within each fundamental period parallelogram of the elliptic function ${\widetilde{\mathcal{P}}}^m(x,\cdot)$, the planar density is the reciprocal of the enclosed area, which is readily calculated as a function of~$x$ (see \cite[equation~(4.144)]{BuckinghamM14}). The linear density is similarly the reciprocal of the length of the period interval, since for $x\in T\cap\mathbb{R}$ all poles are real (modulo the period lattice). This leads to the explicit formula
\begin{gather*}
{\widetilde{\sigma}}_\mathrm{L}(x)=\left[2\int_{D(x)}^{A(x)}\frac{\dd \kappa}{R(\kappa;x)} + 2\int_{D(x)}^{B(x)}\frac{\dd \kappa}{R(\kappa;x)}\right]^{-1}>0,\qquad x\in T\cap\mathbb{R}.
\end{gather*}
While the planar and linear densities are def\/ined here from the known approximation ${\widetilde{\mathcal{P}}}^m(x,w)$, they indeed capture the true local densities of poles of $p_m(m^{2/3}x)$ \cite[Theorem~5]{BuckinghamM14} in the limit of large~$m$.

Another type of result aims to capture the ``local average'' behavior of $p_m(y)$. Here one notes that as $p_m(y)$ has simple poles only, it is locally integrable with respect to area measure in the plane. Similarly, integrals of $p_m(y)$ with respect to Lebesgue measure on $\mathbb{R}$ are well-def\/ined if interpreted in the principal-value sense. Thus, the following local averages are well-def\/ined for $x\in T$ and $x\in T\cap\mathbb{R}$ respectively:
\begin{gather*}
\big\langle{\widetilde{\mathcal{P}}}\big\rangle(x):=\frac{\iint_{\mathfrak{p}(x)}{\widetilde{\mathcal{P}}}^m(x,w)\,\dd A(w)}{\iint_{\mathfrak{p}(x)}\,\dd A(w)},\qquad x\in T,
\end{gather*}
where $\mathfrak{p}(x)$ denotes a period parallelogram and $\dd A(w)$ is area measure in the $w$-plane, and
\begin{gather*}
\big\langle{\widetilde{\mathcal{P}}}\big\rangle_\mathbb{R}(x):=\frac{1}{L}\mathrm{P.V.}\int_{w_0}^{w_0+L}{\widetilde{\mathcal{P}}}^m(x,w)\,\dd w,\qquad x\in T\cap\mathbb{R},
\end{gather*}
where $L$ is the length of a real period interval and $w_0$ is not a pole of the integrand. Remarkably, as shown in \cite[Proposition~11]{BuckinghamM14}, these two quite dif\/ferent def\/initions actually agree where both are def\/ined:
\begin{gather*}
\big\langle{\widetilde{\mathcal{P}}}\big\rangle_\mathbb{R}(x)=\big\langle{\widetilde{\mathcal{P}}}\big\rangle(x),\qquad x\in T\cap\mathbb{R}.
\end{gather*}
Also, $\langle{\widetilde{\mathcal{P}}}\rangle(x)$ can be expressed in terms of basic quantities associated with the Riemann sur\-fa\-ce~$\Gamma(x)$. It is furthermore shown in \cite[Proposition~12]{BuckinghamM14} that $\langle{\widetilde{\mathcal{P}}}\rangle(x)$ may be extended to the whole complex $x$-plane as a continuous function by def\/ining $\langle{\widetilde{\mathcal{P}}}\rangle(x):={\widetilde{\mathcal{P}}}_0(x)$ (the distinguished solution of the cubic equation~\eqref{eq:cubic}) for $x\in\mathbb{C}\setminus T$. This extended function is analytic in $x$ outside of~$T$ but fails to be analytic within~$T$. Then we have the following result.
\begin{Theorem}[Buckingham \& Miller \protect{\cite[Corollary~3 \& Theorem~4]{BuckinghamM14}}]\label{thm:weaklimit}
\begin{gather*}
\lim_{m\to\infty} m^{-1/3}p_m\big(m^{2/3}{\diamond}\big) = \big\langle{\widetilde{\mathcal{P}}}({\diamond})\big\rangle,
\end{gather*}
where the convergence is in the sense of the distributional topology on $\mathscr{D}'(\mathbb{C}\setminus\partial T)$. Also if $\varphi\in \mathscr{D}((\mathbb{C}\setminus\partial T)\cap\mathbb{R})$ is a smooth test function with compact real support avoiding $\partial T$, then
\begin{gather*}
\lim_{m\to\infty}\mathrm{P.V.}\int_\mathbb{R} m^{-1/3}p_m\big(m^{2/3}x\big)\varphi(x)\,\dd x = \int_\mathbb{R}
\big\langle{\widetilde{\mathcal{P}}}\big\rangle(x)\varphi(x)\,\dd x,
\end{gather*}
expressing a similar distributional convergence where the integrals have to be interpreted in the principal value sense.
\end{Theorem}

\subsubsection[Asymptotic description of $p_m$ near edges]{Asymptotic description of $\boldsymbol{p_m}$ near edges}
The function $\mathfrak{d}(x):=\mathfrak{c}(x)-\ii\pi/2$ (cf.~\eqref{eq:c-define}) turns out to be a conformal mapping on a neighborhood of any sub-arc of the edge of $\partial T$ that crosses the positive real $x$-axis, and it maps this edge onto the imaginary segment with endpoints $\pm\ii\pi/2$. Also recalling the function $r(\kappa;x)$ from Section~\ref{sec:elliptic-region}, let $r_*(x):=r\big({\widetilde{\mathcal{P}}}_0(x);x\big)$ and def\/ine
\begin{gather*}
\ell(x):=-\frac{1}{2}\log\big(9r_*(x)^5{\widetilde{\mathcal{P}}}_0(x)\big)
\end{gather*}
to be real for $x\in\partial T\cap\mathbb{R}_+$ and analytically continued to the neighborhood of the sub-arc in question. Denoting by $h_n$ the leading coef\/f\/icient of the normalized Hermite polynomial:
\begin{gather*}
h_n:=\frac{2^{n/2}}{\pi^{1/4}\sqrt{n!}},\qquad n=0,1,2,3,\dots,
\end{gather*}
we def\/ine inf\/initely many complex coordinates (shifts of $\mathfrak{d}(x)$) by
\begin{gather*}
X_n^m(x):=\mathfrak{d}(x)+\tfrac{1}{2}\big(n+\tfrac{1}{2}\big)\frac{\log\big(m-\tfrac{1}{2}\big)}{m-\tfrac{1}{2}}
-\frac{n+\tfrac{1}{2}}{m-\tfrac{1}{2}}\ell(x)+\frac{\log\big(\sqrt{2\pi}h_n\big)}{m-\tfrac{1}{2}},\!\!\!\qquad n=0,1,2,3,\dots.
\end{gather*}
Finally, def\/ine the trigonometric functions $T_n^m(x)$ by
\begin{gather*}
T^m_n(x):=\begin{cases}
1+\coth\big(\big(m-\tfrac{1}{2}\big)X_n^m(x)\big), &n\equiv m\pmod{2},\\
1+\tanh\big(\big(m-\tfrac{1}{2}\big)X_n^m(x)\big), &n\not\equiv m\pmod{2},\quad n=0,1,2,3,\dots.
\end{cases}
\end{gather*}
Then we have the following result.
\begin{Theorem}[Buckingham \& Miller \protect{\cite[Theorem~2]{BuckinghamM15}}]\label{thm:edge} Let arbitrarily small constants $\delta>0$ and $\sigma>0$, and an arbitrarily large constant $M>0$ be given. Suppose that $\operatorname{Re}(\mathfrak{d}(x))\ge -M\log(m)/m$ and $|\arg(x)|\le\pi/3-\sigma$ $($this puts $x$ in the sector containing the edge of $\partial T$ of interest and prevents~$x$ from penetrating the elliptic region~$T$ by a distance greater than~$\mathcal{O}(\log(m)/m))$. Suppose also that $x$ is of distance at least $\delta/m$ from every pole of the functions $T_n^m(x)$, $n=0,1,2,3,\dots$. Then
\begin{gather*}
m^{-1/3}p_m\big(\big(m-\tfrac{1}{2}\big)^{2/3}x\big)\\
{}={\widetilde{\mathcal{P}}}_0(x) + \sum_{n=0}^\infty
\left[-\frac{1}{2}r_*(x)T^m_n(x)+\frac{3{\widetilde{\mathcal{P}}}_0(x)r_*(x)(r_*(x)
-2{\widetilde{\mathcal{P}}}_0(x))^2T^m_n(x)}{6{\widetilde{\mathcal{P}}}_0(x)r_*(x)(r_*(x)-2{\widetilde{\mathcal{P}}}_0(x))T^m_n(x)-4}\right] + \mathcal{O}\big(m^{-1}\big)
\end{gather*}
holds as $m\to\infty$ uniformly for the indicated $x$.
\end{Theorem}
Note that the inf\/inite series is easily seen to be convergent, and the whole series decays rapidly to zero as $m\to\infty$ if $x$ lies outside of $T$, in which case this result agrees with Theorem~\ref{thm:outside}. As~$x$ enters $T$, the terms in the series ``turn on'' one at a time, producing the curves of poles roughly parallel to the edge as can be seen in Fig.~\ref{fig:poles}. Note that $T^m_n(x)=\mathcal{H}_{mn}(x)+1$ and ${\widetilde{\mathcal{P}}}_0(x) =-\tfrac{1}{2}S(x)$ in the notation of~\cite{BuckinghamM15}. One can observe from Theorem~\ref{thm:edge} that the curves of poles roughly correspond to the straight vertical lines $\operatorname{Re}(\mathfrak{d}(x))=-\tfrac{1}{2}\big(n+\tfrac{1}{2}\big)\log(m)/m$ in the $\mathfrak{d}$-plane. There is also an interesting vertical ``staggering'' ef\/fect of the pole lattice as $m$ varies. Indeed, given a value of $\alpha\in \big({-}\tfrac{1}{2},\tfrac{1}{2}\big)$, the poles of the approximation formula near the line indexed by~$n$ with $|\operatorname{Im}(\mathfrak{d}(x))-\pi\alpha|=\mathcal{O}(m^{-1})$ form an approximate vertical lattice in the $\mathfrak{d}$-plane with spacing~$\ii\pi/m$. The lattice is of\/fset from the point $\mathfrak{d}=\ii\pi\alpha-\tfrac{1}{2}\big(n+\tfrac{1}{2}\big)\log(m)/m$ by a complex shift proportional to~$m^{-1}$ (i.e., proportional to the spacing) and depending on~$m$,~$n$, and~$\alpha$. Holding~$m$ f\/ixed, one can observe that near the real axis this of\/fset changes by approximately half of the lattice spacing with each consecutive value of $n$, and as~$x$ moves along the edge toward the corner in the upper half-plane, this change in the of\/fset with $n$ gradually increases to approximate\-ly~$3/4$ of the spacing. On the other hand, holding~$n$ f\/ixed and therefore looking just at the poles along the $n^\mathrm{th}$ line from the edge, the change in of\/fset with~$m$ is again half of the spacing near the real axis, but now the ef\/fect diminishes to zero as one moves along the edge toward a corner of~$\partial T$. This latter ef\/fect implies, in as much as one can draw conclusions from Theorem~\ref{thm:edge} in the situation that~$x$ approaches a corner point along an edge, the pattern of poles of $p_m(y)$ should become independent of $m$ near a corner point, even though it f\/luctuates wildly near typical points of~$T$. A more precise version of this observation will be discussed in Section~\ref{sec:corner}.

\subsubsection[Asymptotic description of $p_m$ near corners]{Asymptotic description of $\boldsymbol{p_m}$ near corners}\label{sec:corner}
The Painlev\'e-I equation $Y''(t)=6Y(t)^2+t$ has a unique \emph{tritronqu\'ee} solution with the property that
\begin{gather}
Y(t)=-\left(\frac{t}{6}\right)^{1/2}+\mathcal{O}\big(t^{-2}\big),\qquad t\to\infty,\qquad |\arg(-t)|\le\frac{4}{5}\pi-\delta\label{eq:tritronquee-asymptotics}
\end{gather}
for every $\delta>0$; see Kapaev \cite{Kapaev04}. Thus the tritronqu\'ee solution $Y(t)$ is asymptotically pole-free in a sector of opening angle $4\pi/5$. It has recently been proven \cite{CostinHT14} that in fact $Y(t)$ is exactly pole-free for $|\arg(-t)|\le 4\pi/5$ without any condition on $|t|$. This is the particular solution of the Painlev\'e-I equation appearing in the formal analysis described in Section~\ref{sec:heuristics} that is needed to describe the rational Painlev\'e-II functions near corner points of $T$ as the following result shows. Recall that $x_\mathrm{c}:=-(9/2)^{2/3}$ is the corner point of $T$ on the negative real axis.
\begin{Theorem}[Buckingham \& Miller \protect{\cite[Theorem~3]{BuckinghamM15}}]\label{thm:corner} Let $Y(t)$ be the tritronqu\'ee solution of the Painlev\'e-I equation determined by the asymptotic expansion \eqref{eq:tritronquee-asymptotics}. If $\mathcal{K}$ is any compact set in the complex $t$-plane that does not contain any poles of $Y(t)$, then
\begin{gather*}
m^{-1/3}p_m\big(\big(m-\tfrac{1}{2}\big)^{2/3}x\big)=-6^{-1/3}-\frac{1}{m^{2/5}}\left(\frac{128}{243}\right)^{1/15}Y(t)+\mathcal{O}\big(m^{-3/5}\big)
\end{gather*}
holds as $m\to\infty$ uniformly for
\begin{gather*}
t:=\left(\frac{2}{243}\right)^{1/15}m^{4/5}(x-x_\mathrm{c})\in\mathcal{K}.
\end{gather*}
\end{Theorem}
This result is interesting in part because $p_m(y)$ is a function with simple poles only, and the approximating function $Y(t)$ is known to have double poles only. What actually happens in the limit $m\to\infty$ near the corner points is that pairs of simple poles of opposite residue for $p_m(y)$ merge into the ``holes'' excluded from $\mathcal{K}$ located near the double poles of $Y$. This phenomenon can be clearly observed in the plots shown in~\cite{BuckinghamM15}. The ``pairing'' of poles of opposite residues near the corners can also be seen in Fig.~\ref{fig:corner-pairs}.
\begin{figure}[t]\centering
\includegraphics[width=0.3\linewidth]{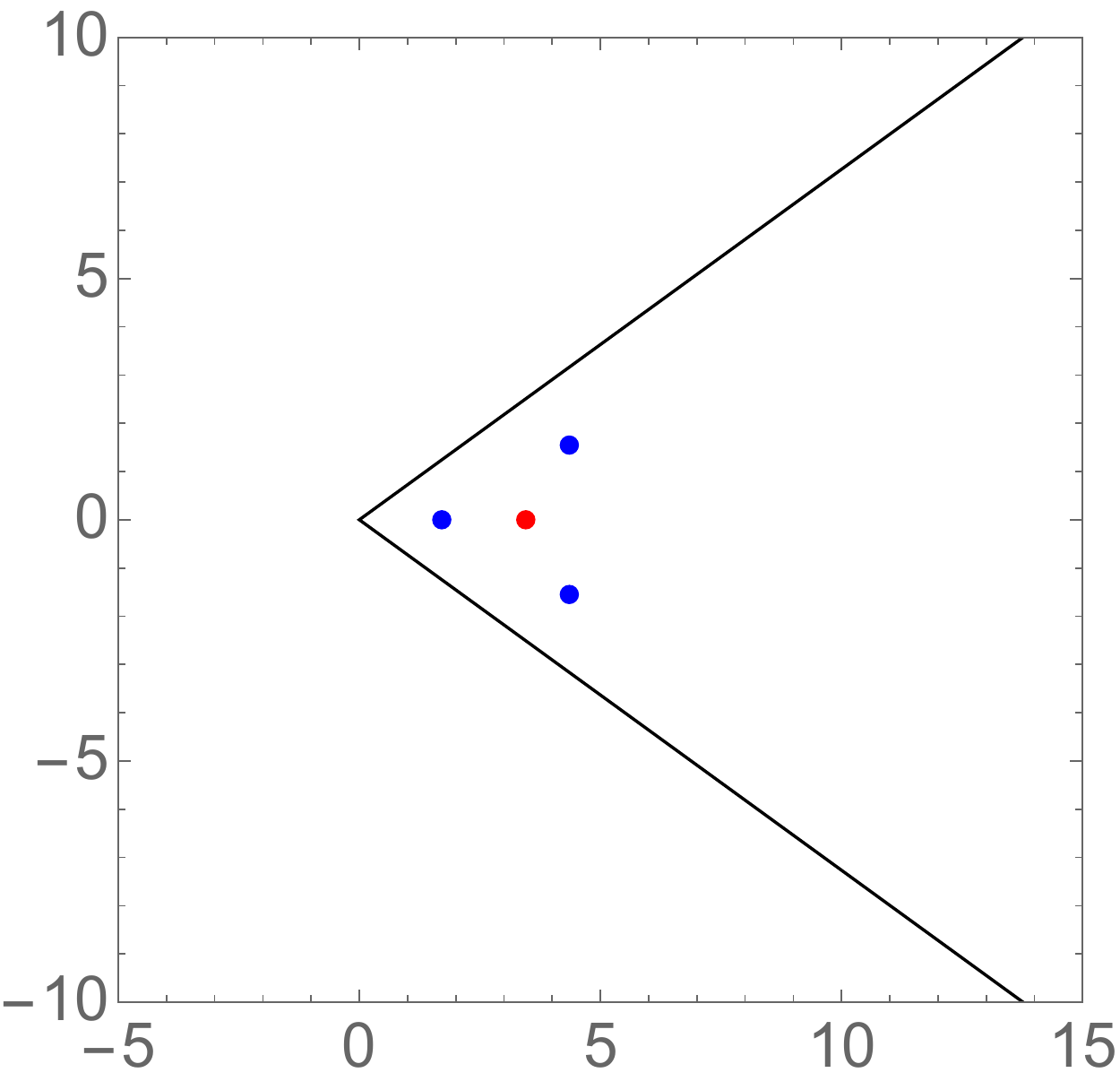}\hspace{0.03\linewidth}%
\includegraphics[width=0.3\linewidth]{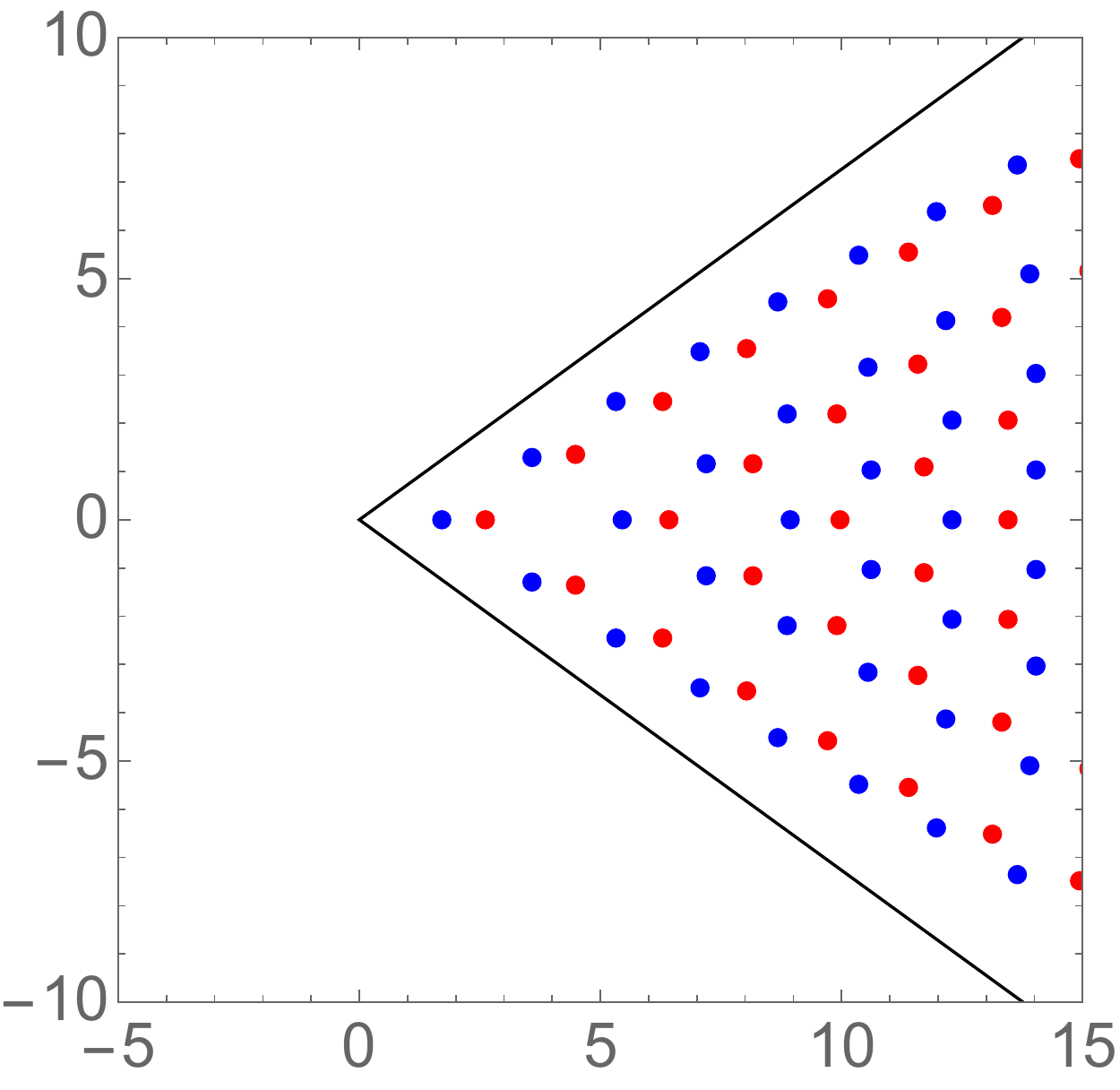}\hspace{0.03\linewidth}%
\includegraphics[width=0.3\linewidth]{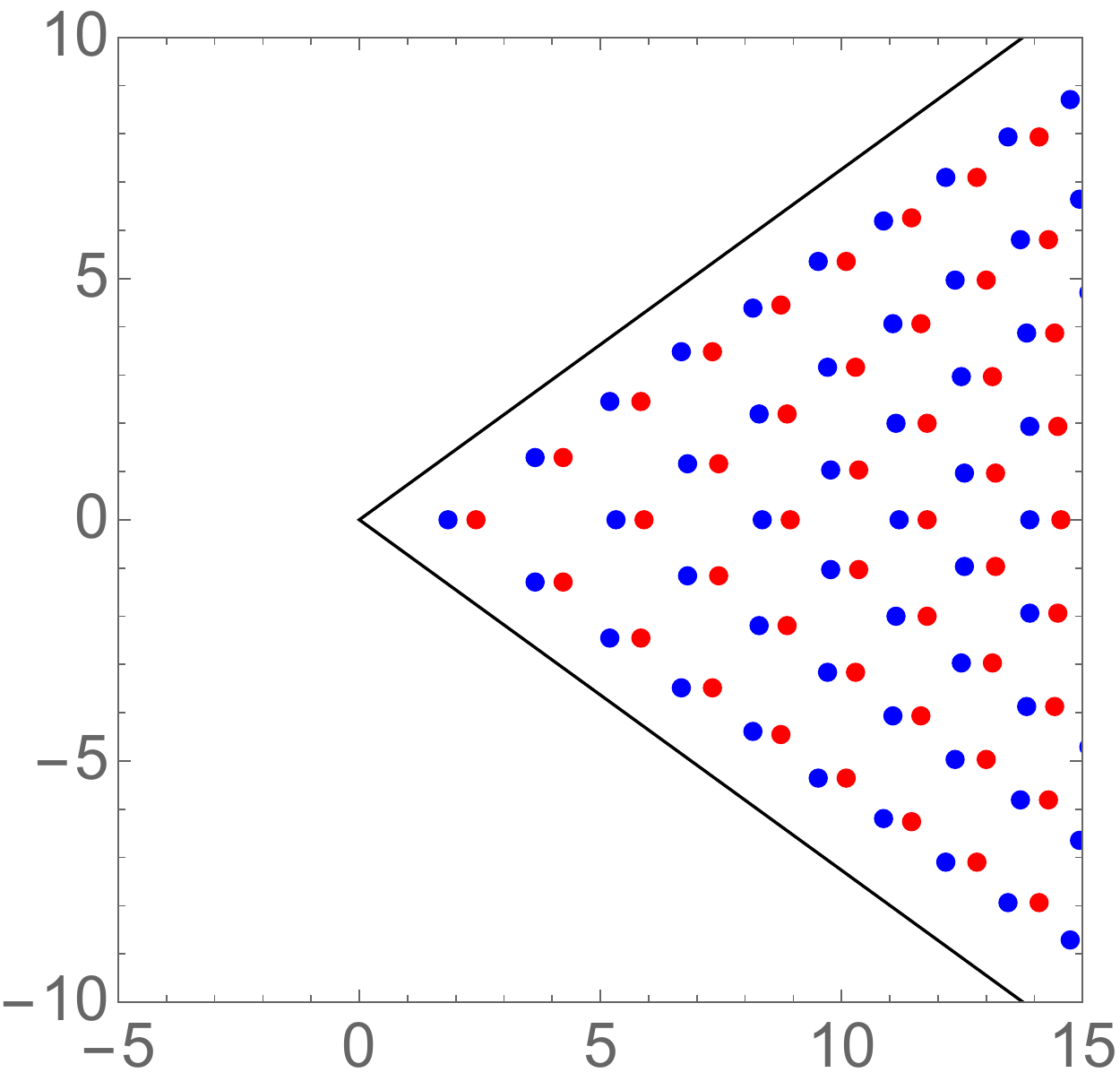}

\caption{The poles of residue $1$ (blue) and $-1$ (red) of $p_{2}(y)$ (left), $p_{11}(y)$ (center), and $p_{60}(y)$ (right), plotted in the complex $t$-plane, along with the boundary $|\arg(t)|=\pi/5$ of the pole sector for the \mbox{Painlev\'e-I} tritronqu\'ee solution $Y(t)$. Note how as $m$ increases pairs of poles of opposite residues coalesce (each pair moving toward a double pole of $Y(t)$).}\label{fig:corner-pairs}
\end{figure}

Finally, we remark that the careful reader will observe that the various domains of the complex $y$-plane in which the asymptotic behavior of $p_m$ is now known actually do not overlap, so the whole complex plane has not been covered. The uniform asymptotic description of $p_m$ in neighborhoods of the edges and corners of $T$ suf\/f\/iciently large to achieve overlap remains an open technical problem.

\subsection*{Acknowledgements}
P.D.~Miller was supported during the preparation of this paper by the National Science Foundation under grant DMS-1513054. The authors are grateful to Thomas Bothner for many useful discussions.

\pdfbookmark[1]{References}{ref}
\LastPageEnding

\end{document}